\documentclass[aps,prb,twocolumn,superscriptaddress,floatfix,nofootinbib]{revtex4-2}
\usepackage[utf8]{inputenc}
\usepackage{natbib}
\usepackage[colorlinks=true,linkcolor=blue,citecolor=blue,urlcolor=blue]{hyperref}
\usepackage{graphicx}
\usepackage{latexsym}
\usepackage{amssymb}
\usepackage{amsmath}
\usepackage{amsfonts}
\usepackage{amsthm}
\usepackage{bm}
\usepackage{mathtools}
\usepackage[caption=false]{subfig}
\usepackage{bbm}
\usepackage{enumitem}
\usepackage{hyperref}
\usepackage{lipsum}
\usepackage{braket}
\usepackage{tikz}
\usepackage{pgfplots}
\usepackage{comment}
\usepackage{diagbox}
\usepackage{ifthen}
\usepackage{ragged2e}
\usepackage{multirow}
\usepackage[export]{adjustbox}
\usetikzlibrary{arrows,decorations.pathreplacing,decorations.markings,arrows.meta,patterns,3d}
\usepackage[normalem]{ulem} % for strikeout \sout https://tex.stackexchange.com/a/23712
\usepackage{cancel}
\usepackage{microtype}
\usepackage{booktabs}

\theoremstyle{definition}
\newtheorem{proposition}{Proposition}
\newtheorem{lemma}{Lemma}
\newtheorem{claim}[proposition]{Claim}

\newcommand\CD{\mathcal{D}}

\pgfplotsset{compat=1.8}

\begin{document}

\title{Sequential Circuits as Generalized Symmetry on the Lattice}
\date{\today}

\author{Nathanan Tantivasadakarn}
\affiliation{Walter Burke Institute for Theoretical Physics,}
\affiliation{Department of Physics and Institute for Quantum Information and Matter, \mbox{California Institute of Technology, Pasadena, CA, 91125, USA}}
\affiliation{C. N. Yang Institute for Theoretical Physics, Stony Brook University, Stony Brook, NY 11794, USA}

\author{Xinyu Liu}
\affiliation{Department of Physics and Institute for Quantum Information and Matter, \mbox{California Institute of Technology, Pasadena, CA, 91125, USA}}

\author{Xie Chen}
\affiliation{Walter Burke Institute for Theoretical Physics,}
\affiliation{Department of Physics and Institute for Quantum Information and Matter, \mbox{California Institute of Technology, Pasadena, CA, 91125, USA}}

\begin{abstract} 
Generalized symmetry extends the usual notion of symmetry to ones that are of higher-form, acting on subsystems, non-invertible, etc. The concept was defined in the field theory context using the idea of topological defects. On the lattice, an immediate consequence is that a symmetry twist is moved across the system by a sequential quantum circuit. In this paper, we ask how to deduce the full, potentially non-invertible symmetry action from the unitary sequential circuit and how the connection to sequential circuit constrains the properties of the generalized symmetries. We find that for symmetries that contain the trivial symmetry operator as a fusion outcome, which we call annihilable symmetries, the sequential circuit fully determines the symmetry action and puts various constraints on their fusion. In contrast, for unannihilable symmetries, like those whose corresponding twist is the Cheshire string, a further one-dimensional sequential circuit is needed for the full description. Matrix product operator and tensor network operator representations play an important role in our discussion.
\end{abstract}

\maketitle

%%%%%%%%%%%%%%%%%%%%%%%%%%%%%%%%%%%%%%%%%%%%%%%%%%%%%%%%%%%%%%%%%%%

\section{Introduction}

Symmetry -- arguably one of the most important concepts in physics -- has recently gained new life in the context of quantum many-body systems. The usual types of symmetries discussed in the quantum many-body setting -- spin rotation, time reversal, lattice translation -- are unitary / anti-unitary operators that act globally on the whole system with symmetry charges carried by point-like operators. Each of these aspects can break down in the now generalized version of what is considered a symmetry\cite{Gaiotto2015,McGreevy2023,SCHAFERNAMEKI2024,Shao2024TASI}. A generalized symmetry does not need to act on the whole system but instead can be restricted to certain subsystems. The symmetry charge can be carried by extended objects like a line or a membrane. Moreover, a generalized symmetry operation may not have an inverse like (anti-)unitary operators do and is called `non-invertible'. 

The notion of generalized symmetry was proposed and has been extensively studied in the field theory context. Formal mathematical concepts like the (higher) category theory has been applied to reveal their intricate structures. In the condensed matter community, the notion has drawn some interest especially with the proposal of concrete lattice models realizing these exotic symmetries\cite{Feiguin2007,Inamura2022,GarreRubio2023,Fechisin25,Seifnashri2024,Bhardwaj2025,Bhardwaj2025b}. More connections need to be made to reveal the full power of generalized symmetries in the context of condensed matter, which could lead to a generalized Landau paradigm of quantum phase and phase transitions\cite{McGreevy2023,Bhardwaj2023Landau}. In this paper, we address the question of `\textit{what is a generalized symmetry on the lattice}'. In answering it, we make use of a special type of quantum circuit -- sequential quantum circuits. 

It is interesting to think about symmetry from the perspective of quantum circuits. Most of the conventional (anti-)unitary global symmetries are implemented in an `on-site' way as a tensor product of local (anti-)unitary operators. For example, the spin flip in a spin-$1/2$ system is implemented as
\begin{equation}
\sigma^1_x \otimes \sigma^2_x \otimes ... \otimes \sigma^N_x
\end{equation}
Translation or other spatial symmetries involves mapping one local degree of freedom to another, but other than that, it is similar to the internal symmetries in that no entanglement is generated by the symmetry action and local operators retain their size when conjugated by the symmetry operation.

A more exotic type of symmetry is one with an anomaly, for example, the effective symmetry on the boundary of symmetry protected topological (SPT) phase. Take for instance the effective $Z_2$ symmetry on the $1$D boundary of the $2$D $Z_2$ SPT phase\cite{Chen2011,Levin2012}
\begin{equation}
\begin{array}{l}
\left(\alpha^{1,2} \otimes \alpha^{3,4} \otimes ... \otimes \alpha^{2N-1,2N}\right)  \left(\alpha^{2,3} \otimes \alpha^{4,5} \otimes ... \otimes \alpha^{2N,1}\right)
\\
\times \left(\sigma^1_x \otimes \sigma^2_x \otimes ... \otimes \sigma^{2N}_x \right)
\end{array}
\end{equation}
where $\alpha$ is a two-qubit diagonal unitary operator $\alpha = \text{diag}(1,i,i,1)$. The whole operator is unitary but it no longer has a tensor product structure any more. Instead, it can be implemented with a few layers of tensor product unitaries -- a finite depth quantum circuit. A generic product state becomes a many-body entangled state under the action of the unitary, and a local operator can grow in size by a finite amount upon conjugation by the circuit. Related to the non-onsite-ness of the symmetry operator is the fact that this is an anomalous $Z_2$ symmetry and no product state (or any short-range correlated state) is invariant under the symmetry.

Are generalized symmetries related to a more general class of quantum circuits when implemented on lattice? In section~\ref{sec:proof}, we present the argument that generalized symmetries, following the field theory definition,  are implemented as Sequential Quantum Circuits on the lattice (supplemented by non-unitary operations at the boundary). This fact was pointed out in Ref.~\onlinecite{Shao2024TASI} and was used to derive the action of the Kramers-Wannier symmetry on the Ising chain. Sequential Quantum Circuits\cite{Chen2024} are a much more powerful class of circuits than tensor product unitaries or finite depth circuits. They can generate long-range correlations, long-range entanglement, map between different gapped phases and do not necessarily preserve locality. Therefore, when implemented as symmetry operators, they can go beyond conventional symmetries implemented as a tensor product of unitaries or finite depth circuits. In section~\ref{sec:KW} and section~\ref{sec:Cheshire}, we illustrate this power with two examples: the Kramers-Wannier transformation as a non-invertible symmetry in 1D spin chain and the 2D symmetry with the Cheshire string as the symmetry twist acting on a topological state. The former maps between symmetric and symmetry breaking phases of the Ising chain and hence can generate long-range correlation. The latter, as we will show, effectively implements the $1$-form symmetry in a topological state in a $0$-form way. The Kramers-Wannier example is generalized in section~\ref{sec:1D} through the matrix product operator formalism to all 0-form generalized symmetries in 1d where a set of properties is proven starting from the fact that these generalized symmetries are implemented as 1d sequential circuits.

\begin{figure}[th]
\includegraphics[scale=0.65]{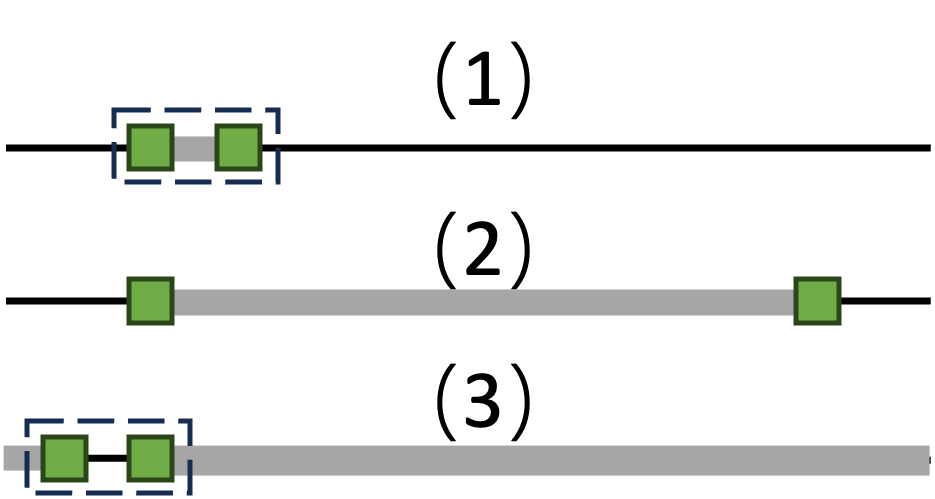}
\caption{Implementation process of a generalized symmetry. (1) symmetry is applied to a sub-region (dashed box) generating symmetry twists (green boxes) on the boundary; (2) a symmetry twist is swept through the bulk of the system; (3) symmetry twists are brought back together and annihilated.}
\label{fig:symmetry}
\end{figure}

An immediate issue that needs to be addressed is how can the generalized symmetry be non-invertible if they are implemented as unitary quantum circuits. To understand how this is not contradictory, we need to look closer into the process a symmetry is implemented. As illustrated in Fig.~\ref{fig:symmetry}, we envision a process where the symmetry is first applied to a sub-region of the whole system. As long as the Hamiltonian of the system consists of local symmetric terms, only terms on the boundary of the sub-region might get changed, inducing the so-called `symmetry twist' on the boundary (step 1). The symmetry twist is then moved (swept) through the bulk of the system as the symmetry is applied to bigger and bigger sub-regions (step 2). Finally, the symmetry twists are annihilated and the symmetry action covers the whole system (step 3). In our argument in section~\ref{sec:proof}, we focus only on step (2) of sweeping the symmetry twist through the bulk of the system without addressing how the symmetry twists are created and annihilated. While step (2) is implemented with a sequential quantum circuit, step (1) and (3) can involve non-unitary operations, resulting in a non-invertible symmetry action. Interestingly, as we show using the examples in
section~\ref{sec:KW} and section~\ref{sec:Cheshire}, the non-unitary part of the
operation can be deduced once we know all the sequential
circuits involved in the implementation (possibly in step (1) and (3) as well) and no extra input is needed. 

\section{Generalized Symmetry as Sequential Quantum Circuit}
\label{sec:proof}

In this section, we present the general argument for why a generalized symmetry -- as defined in Ref.~\cite{Gaiotto2015,Bhardwaj2023GS, Shao2024TASI}-- is implemented as a Sequential Quantum Circuit on the lattice. 

\begin{figure}[th]
\includegraphics[scale=0.35]{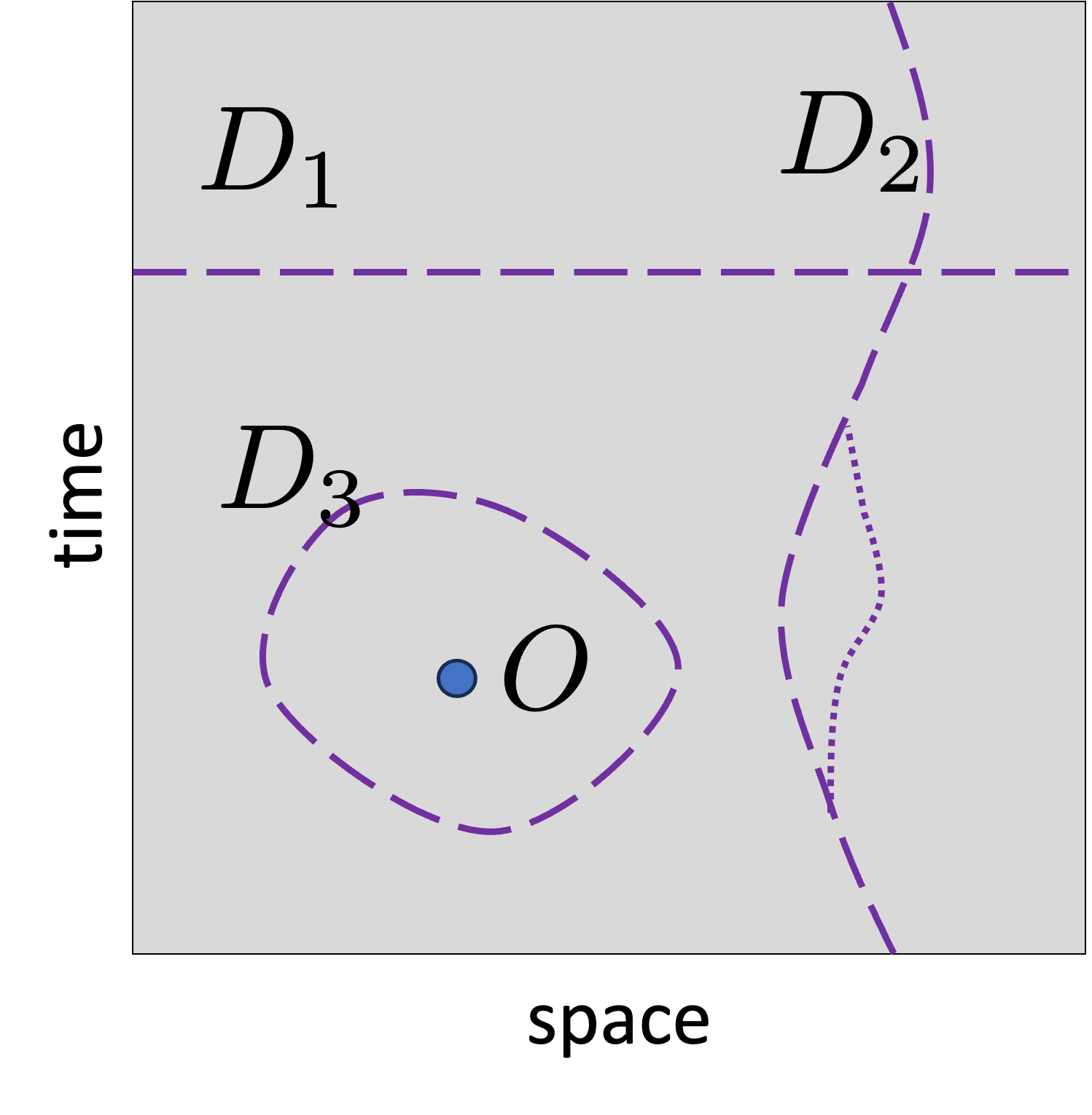}
\caption{Topological defects in a field theory. A defect (a dashed purple line) can extend along the spatial direction ($D_1$), the time direction ($D_2$), or act locally ($D_3$). A defect is called `topological' if smooth deformation of it does not change any correlation function.}
\label{fig:TopoDefect}
\end{figure}

In the field theory context, generalized symmetries are defined based on the notion of topological defects. Fig.~\ref{fig:TopoDefect} illustrates the situation in $1+1$D space time. Dashed lines represent defects in the path integral. A defect is called `topological' if the correlation function represented by the path integral remains invariant under smooth deformations of the defect, as long as the deformation does not pass through any inserted operators $O$ (for example from the dashed to the dotted $D_2$). When a topological defect extends purely along the spatial direction ($D_1$), it becomes the operator that applies the corresponding generalized symmetry to the underlying Hilbert space. When the topological defect extends in the time direction ($D_2$), it changes the Hilbert space by inserting a symmetry twist.

\begin{figure}[th]
\includegraphics[scale=0.35]{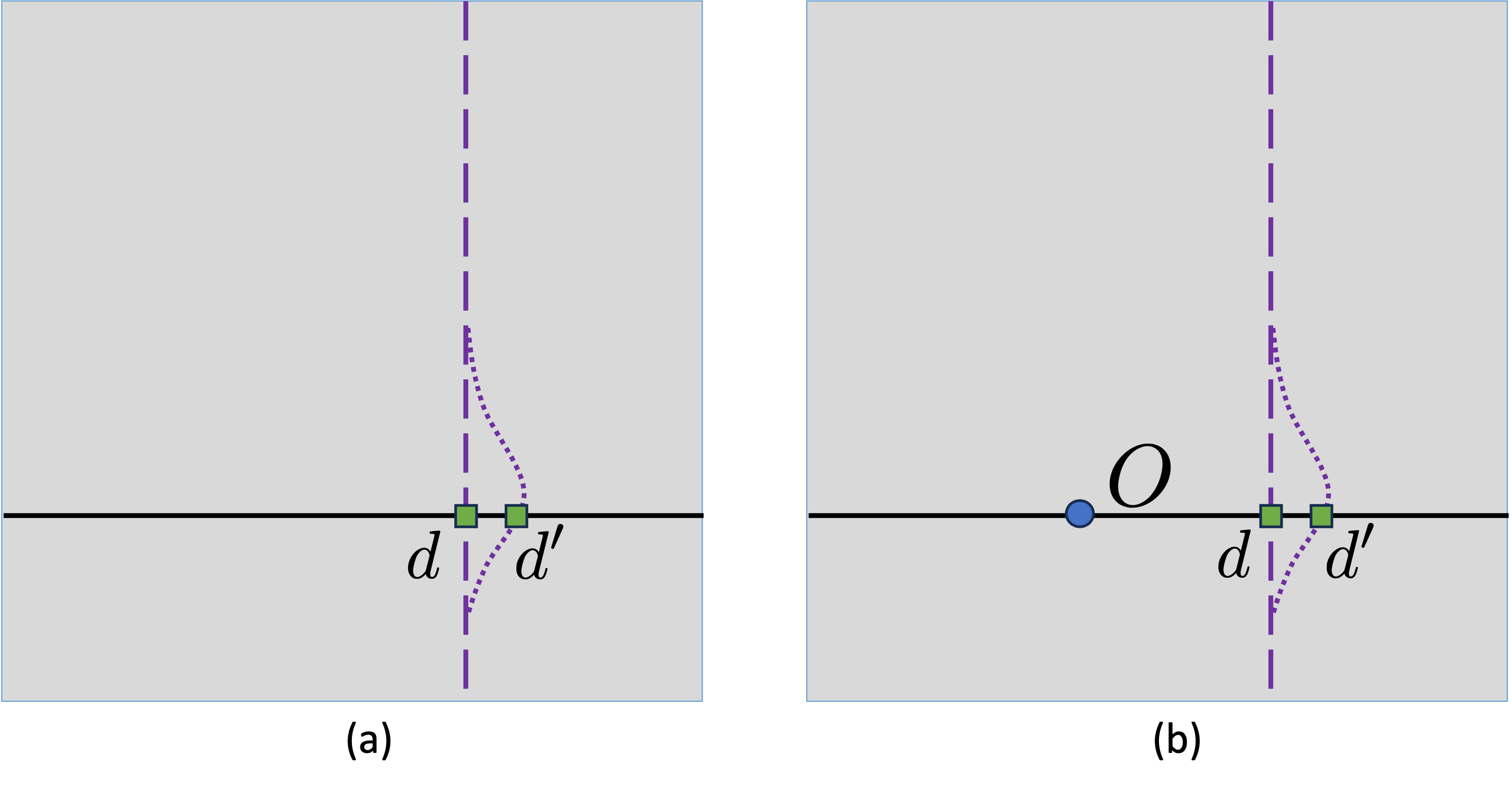}
\caption{Deforming a time direction topological defect (the dashed purple line) in a $1+1$d system and the resulting movement of the symmetry twist (green box) from $d$ to $d'$ in the absence (a) or presence (b) of other operators $O$. The black line indicates a cut in the spatial direction and exposes the underlying Hilbert space.}
\label{fig:Defect2Circuit}
\end{figure}

Now we argue that this definition directly translates into a sequential circuit when we go to the Hamiltonian formulation. Specifically, we will show that the movement of a symmetry twist from one location to a nearby location can be achieved by a local unitary transformation between the two locations. Therefore, sweeping the symmetry twist through the whole system to implement the full symmetry can be achieved with a sequential circuit composed of such local unitary steps. This statement applies when the symmetry twist is a point-like object. In higher dimensions ($2+1$d and higher), a symmetry twist can be an extended object like a line or a membrane. We will show that local deformations of a symmetry twist can be achieved with a local unitary around the location of the deformation. The sweeping of the whole twist is then achieved with a sequential circuit consisting of finite-depth circuit steps along the dimension of the twist. If the symmetry under consideration is a subsystem symmetry, the argument applies within the sub-manifold where the symmetry acts.

First, we focus on the case of point-like symmetry twists in $1+1$d systems. The discussion can be straight-forwardly generalized to point-like twists in higher dimensions. To expose the Hilbert space in the Hamiltonian formulation, we cut the space-time open along a spatial slice, as shown in Fig.~\ref{fig:Defect2Circuit} with the solid lines. The intersection between the spatial slice and a defect line, as indicated by a square box, is a symmetry twist. When there is a single defect that extends in the time direction, the Euclidean path integral gives the partition function of the system with a symmetry twist. 
\begin{equation}
Z \sim \text{Tr}\left(e^{-\beta H_d} \right)
\end{equation}
The topological nature of the defect indicates that the path integral remains invariant when the symmetry twist moves along the spatial slice (from $d$ to $d'$ in Fig.~\ref{fig:Defect2Circuit} (a)), if no other operators are inserted along the path of the movement. Since this holds for any $\beta$, $H_d$ and $H_{d'}$ must have the same spectrum and be related by a unitary transformation.
\begin{equation}
H_{d'} = U H_d U^{\dagger}
\end{equation}
The topological nature of the defect also requires that correlation functions of any operator $O$
\begin{equation}
\langle O \rangle = \text{Tr} \left(Oe^{-\beta H_d}\right)
\end{equation}
remain invariant if $O$ is outside of the interval from $d$ to $d'$. Therefore, $U$ must commute with $O$. Since $O$ can be any operator outside of the interval between $d$ and $d'$, $U$ must be a local unitary operator between $d$ and $d'$. We hence reach the conclusion that moving a symmetry twist by a short distance is achieved with a local unitary transformation. To sweep a symmetry twist through the whole system and implement the full symmetry action requires then a sequential quantum circuit composed of such local unitary steps. 

\begin{figure}[th]
\includegraphics[scale=0.35]{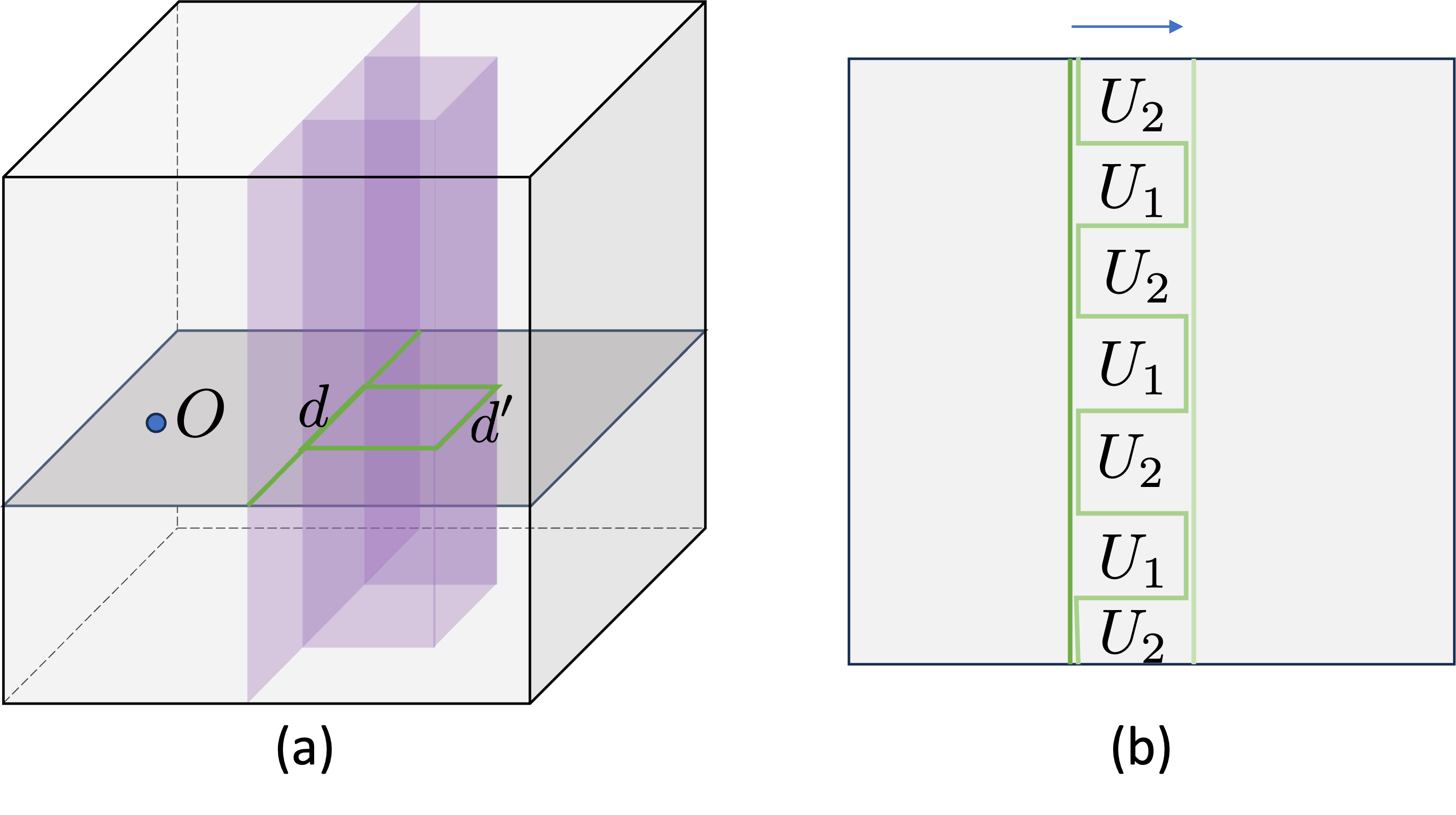}
\caption{(a) Deforming a time direction topological defect (the purple membrane) in a $2+1$d system and the resulting deformation of the symmetry twist (the green line) from $d$ to $d'$. (b) On the 2d spatial slice, moving a 1d symmetry twist (green line) by a finite distance can be achieved with a finite depth 1d circuit (applying all $U_1$'s in parallel and then applying all $U_2$'s in parallel).}
\label{fig:Defect2Circuit2}
\end{figure}

When the symmetry twist is itself an extended object, we need to consider local deformations of the twist before considering their overall movement. Fig.~\ref{fig:Defect2Circuit2} (a) illustrates the local deformation of a $1$d symmetry twist in a $2+1$d system. Again due to the topological nature of the defect and the corresponding invariance of correlation functions, we can conclude that local deformation of the symmetry twist can be achieved with a local unitary transformation at the location of the deformation. The overall movement of a symmetry twist can then be decomposed into two (or any finite number of) layers of commuting local unitary transformations, as shown in Fig.~\ref{fig:Defect2Circuit2} (b). Moving the symmetry twist by a finite distance is hence achieved with a finite depth quantum circuit. To sweep the line-like symmetry twist across the whole system, we need a sequential circuit composed of a sequence of finite depth circuits. 

Of course, as discussed in the introduction, the sequential circuit only covers step (2) of the implementation process. Step (1) and step (3) which creates the annihilates the symmetry twists can contain non-unitary operations.  Can we obtain the full symmetry action (including step (1) and (3)) starting from the sequential circuit in step (2)? In section~\ref{sec:KW} and section~\ref{sec:Cheshire} we discuss two very different cases. 

For the Kramers-Wannier transformation discussed in section~\ref{sec:KW}, the full non-invertible action of the symmetry can be obtained from the sequential circuit part with no extra input. We demonstrate how this can be achieved using the matrix product operator representation of the sequential circuit. Such a construction applies in general to 1D generalized symmetries, as discussed in section~\ref{sec:1D}. This is a natural consequence of the fact that the $0$d symmetry twists of $1$d symmetries can always be (pair)-generated from vacuum with local unitary transformations.  Through the matrix product representation of the symmetry, we can also see how this results in the 1d non-invertible symmetries having fusion rules of the form
\begin{equation}
\label{eq:annihilable}
\mathcal{D}^{\dagger} \times \mathcal{D} = I + ...
\end{equation}
That is, even though these symmetries are non-invertible, there exists an operator $\mathcal{D}^{\dagger}$ such that their fusion contains the identity as one of the fusion results. We call such symmetries `annihilable'.

The situation is very different for the Cheshire string example discussed in section~\ref{sec:Cheshire}. The Cheshire string as a symmetry twist cannot be (pair)-generated from vacuum with a $1$d finite depth quantum circuit. They can only be generated using a $1$d sequential circuit. Therefore, in this case, to obtain the full symmetry action, we need to supplement the $2$d sequential circuit in step (2) with $1$d sequential circuits in step (1) and (3). Each step can further contain non-unitary operations which, as we show in section~\ref{sec:Cheshire}, can be obtained from the sequential circuits with no extra input. We will see that the resulting symmetry action (let us denote it by $\mathcal{C}$) does nothing more than enforcing the $1$-form symmetry on a topological state. Therefore, this provides in a sense a $0$-form implementation of the $1$-form symmetry. The fact that the symmetry twists cannot be (pair) generated from the vacuum with finite depth circuits is directly related to the fact that there does not exist an operator $\mathcal{C}^{\dagger}$ such that its fusion with $\mathcal{C}$ has identity as one of the fusion channels. 
\begin{equation}
\mathcal{C}^{\dagger} \times \mathcal{C} \sim \mathcal{C}
\end{equation}
$\sim$ indicates that we are not very careful with the fusion coefficient. Here, $\mathcal{C}$ is, in a way, more non-invertible than the Kramers-Wannier symmetry $\mathcal{D}$. We call such non-invertible symmetries , i.e., those that do not satisfy Eq.~\eqref{eq:annihilable} `unannihilable'.

\section{Kramers-Wannier Duality as 1D symmetry}
\label{sec:KW}

The Kramers-Wannier duality~\cite{Kramer1941} $\mathcal{D}$ acts on a $1$d chain of spin $1/2$'s with a $Z_2$ global symmetry $\eta = \prod_i X_i$, with $X_i$ being the Pauli-$X$ operator on the $i$-th spin. The $Z_2$ symmetric local operators are mapped under $\mathcal{D}$ as
\begin{equation}
X_i \xrightarrow{\mathcal{D}} Z_iZ_{i+1}, \ Z_iZ_{i+1} \xrightarrow{\mathcal{D}} X_{i+1}
\end{equation}
Therefore, the symmetric phase with Hamiltonian $H = -\sum_i X_i$ is mapped to the symmetry breaking phase with Hamiltonian $H = -\sum_i Z_iZ_{i+1}$ and vice versa, while the critical point at the transition is invariant under the symmetry. 

The Kramers-Wannier transformation is a prototypical example of a non-invertible symmetry which satisfies the fusion rule of
\begin{equation}
\mathcal{D}^{\dagger} \times \mathcal{D} = I + \eta
\end{equation}
It is also known that the transformation can be applied as a sequential circuit followed by a projection to the $Z_2$ symmetric sector. 

The Kramers-Wannier transformation is special in that it has been extensively studied and it is well understood that its non-invertible action is related to the fact that it maps the $Z_2$ symmetry charge to $Z_2$ symmetry twists and vice versa. Therefore, a projection to the no-charge and no-twist sector guarantees that the symmetry is well defined. For general non-invertible symmetries, it may not be immediately obvious what projection or non-unitary transformation is needed to complement the sequential circuit and give the full symmetry action.

We show how the full non-invertible action of the Kramers-Wannier symmetry can be obtained starting from the matrix product operator representation of the sequential circuit. This construction will be applied to all 1D generalized symmetries in section~\ref{sec:1D} and reveal some interesting common feature.

Our proposal is based on the expectation that the sequential quantum circuit should contain all the necessary information for the full description of the Kramers-Wannier transformation (and all 1D generalized symmetries as well). This is because the symmetry twists of Kramers-Wannier (as well as all 1D symmetries) are zero-dimensional objects which we should be able to generate from vacuum using a local unitary operation. Therefore, step 1 in Fig.~\ref{fig:symmetry} could very well be a local patch of the sequential circuit. In step 3, we reserve the process to annihilate the symmetry twists but with a projection to make sure the system end up in the vacuum state. Thus, the sequential circuit should already contain all the information needed to construct the full symmetry action. One feature that is missing in the sequential circuit to describe the full symmetry --- other than the fact that it is invertible --- is that it is not translation invariant. In particular, it maps operators near the ends differently than the ones in the middle. A natural way to resolve this, which at the same time makes the symmetry action potentially non-invertible, is to take the matrix product representation of the bulk part of the sequential circuit and connect copies of it in a translation invariant way. We apply this procedure to $U_{\text{KW}}$ and show that we recover exactly the full non-invertible Kramers-Wannier symmetry derived in Ref.~\onlinecite{Shao2024TASI}.

\begin{figure}[th]
\includegraphics[scale=1.3]{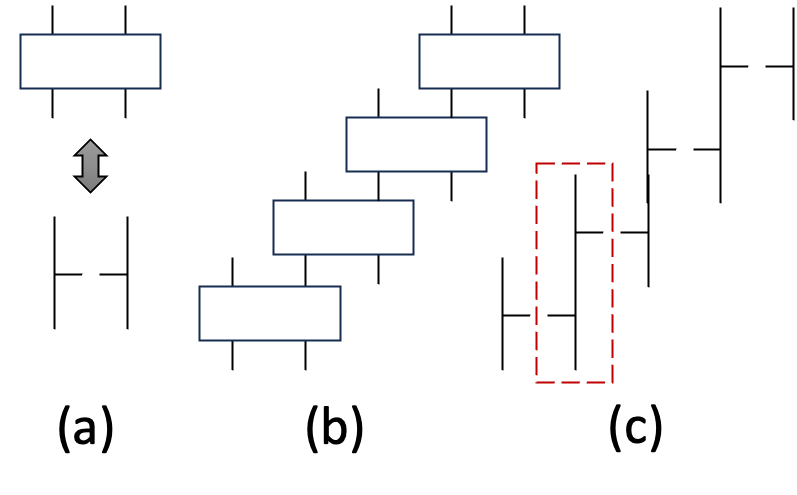}
\caption{Matrix product operator from a sequential circuit. (a) decomposing local unitary gates into tensors; (b) to (c) combining the tensors from each local gate in the sequential circuit gives the matrix product operator representation of the circuit.}
\label{fig:SC2MPO}
\end{figure}

The sequential circuit that implements the Kramers-Wannier transformation is composed of two-body unitaries
\begin{equation}
u_{i,i+1} = e^{i(\pi/4)Z_iZ_{i+1}}e^{i(\pi/4) X_{i}}
\end{equation}
To move the symmetry twist from site $1$ to site $N$, the sequential circuit takes the form
\begin{equation}
U_{\text{KW}} = u_{N-1,N} ... u_{2,3}u_{1,2}
\end{equation}
To obtain the full symmetry action including the non-unitary part, we make sure of the matrix product operator representation of the sequential circuit. As shown in Fig.~\ref{fig:SC2MPO}, each two-body unitary can be decomposed into two rank three tensors and, when they are put together, the sequential circuit can be represented as a matrix product operator. The tensors in the matrix product operator have two physical indices (up and down) and two virtual indices (left and right). In the bulk of the MPO representation, the tensors are the same as long as the two-body unitaries making up the sequential circuit are the same. Given the sequential circuit associated with the Kramers-Wannier transformation, we can derive the MPO representation. 

The two-body unitaries can be decomposed as,
\begin{align}
\label{eq:decomposition}
\sqrt{2} \ \tikz[baseline=-3pt]{\draw[black, thick] (-5,-0.25) rectangle (-3.8,0.25); \draw[black, thick](-4.7,-0.6)--(-4.7,-0.25);\draw[black, thick](-4.1,-0.6)--(-4.1,-0.25);\draw[black, thick](-4.7,0.6)--(-4.7,0.25);\draw[black, thick](-4.1,0.6)--(-4.1,0.25);\node at (-4.4,0){$u_{i,i+1}$};} \to \tikz[baseline=-3pt]{\draw[black, thick](-4.7,-0.6)--(-4.7,0.6);\draw[black, thick](-4.7,0)--(-4.3,0); \node at (-4.1,0){$0$};\node at (-4.5,0.8){$R(X)$};} +\tikz[baseline=-3pt]{\draw[black, thick](-4.7,-0.6)--(-4.7,0.6);\draw[black, thick](-4.7,0)--(-4.3,0); \node at (-4.0,0){$1$};\node at (-4.5,0.8){$iZR(X)$};} 
 \ \  \tikz[baseline=-3pt]{\draw[black, thick](-3.3,-0.6)--(-3.3,0.6);\draw[black, thick](-3.3,0)--(-3.6,0); \node at (-3.8,0){$0$};\node at (-3.3,0.8){$I$};}+ \tikz[baseline=-3pt]{\draw[black, thick](-3.3,-0.6)--(-3.3,0.6);\draw[black, thick](-3.3,0)--(-3.6,0); \node at (-3.8,0){$1$};\node at (-3.3,0.8){$Z$};}
\end{align}
where $R(X) = e^{i\pi/4 X}$. The indices in the vertical direction are input and output physical indices. The indices in the horizontal direction are virtual indices to be contracted to give the physical operator. Using the recombination illustrated in Fig.~\ref{fig:SC2MPO} (c), we find the MPO tensors to be
\begin{equation}
\sqrt{2}M = 
\tikz[baseline=-6pt]{\draw[black, thick] (-0.5,0) -- (0.5,0); \draw[black, thick] (0,-0.5) -- (0,0.5); \node at (-0.65,0){$0$};\node at (0,-0.65){$0$};\node at (0.65,0){$+$};\node at (0,0.65){$0$};} + \tikz[baseline=-6pt]{\draw[black, thick] (-0.5,0) -- (0.5,0); \draw[black, thick] (0,-0.5) -- (0,0.5); \node at (-0.65,0){$0$};\node at (0,-0.65){$0$};\node at (0.65,0){$-$};\node at (0,0.65){$1$};} + \tikz[baseline=-6pt]{\draw[black, thick] (-0.5,0) -- (0.5,0); \draw[black, thick] (0,-0.5) -- (0,0.5); \node at (-0.65,0){$1$};\node at (0,-0.65){$1$};\node at (0.65,0){$+$};\node at (0,0.65){$0$};} - \tikz[baseline=-6pt]{\draw[black, thick] (-0.5,0) -- (0.5,0); \draw[black, thick] (0,-0.5) -- (0,0.5); \node at (-0.65,0){$1$};\node at (0,-0.65){$1$};\node at (0.65,0){$-$};\node at (0,0.65){$1$};} 
\end{equation}
where $|+\rangle = |0\rangle + |1\rangle$, $|-\rangle = |0\rangle - |1\rangle$.

While the sequential circuit is composed of the same two-body unitaries, it is not fully translation invariant due to the boundaries. A fully translation invariant operator can be obtained by taking copies of the tensor $M$ and connect them in a translation invariant way. 
\begin{equation}
\mathcal{D} =  \tikz[baseline=-3pt]{\draw[black, thick] (-0.5,0) .. controls (-0.7,0.0) and (-0.7,0.1) .. (-0.6,0.2);
\draw[black, thick] (-0.5,0) -- (0.5,0); \draw[black, thick] (0,-0.5) -- (0,0.5); \node at (0.2,-0.2){$M$};}\tikz[baseline=-3pt]{\draw[black, thick] (-0.5,0) -- (0.5,0); \draw[black, thick] (0,-0.5) -- (0,0.5);}\tikz[baseline=-3pt]{\draw[black, thick] (-0.5,0) -- (0.5,0); \draw[black, thick] (0,-0.5) -- (0,0.5);}\cdots \cdots \tikz[baseline=-3pt]{\draw[black, thick] (-0.5,0) -- (0.5,0); \draw[black, thick] (0,-0.5) -- (0,0.5);\draw[black, thick] (0.5,0) .. controls (0.7,0.0) and (0.7,0.1) .. (0.6,0.2);}
\end{equation}
The curves at the two ends indicate that the left most index is contracted with the right most one. This gives us the non-invertible operator implementing the full Kramers-Wannier symmetry. This can be seen from the following tensor calculation
\begin{equation}
\mathbb{M} = \tikz[baseline=-3pt]{\draw[black, thick] (-0.5,-0.25) -- (0.5,-0.25); \draw[black, thick] (-0.5,0.25) -- (0.5,0.25);\draw[black, thick] (0,-0.7) -- (0,0.7); \node at (0.2,-0.4){$M$};\node at (0.25,0.45){$M^{\dagger}$};} = \tikz[baseline=-3pt]{\draw[black, thick] (-0.4,-0.25) -- (-0.2,-0.25) -- (-0.2,0.25) -- (-0.4,0.25); \draw[black, thick] (0.4,0.25) -- (0.2,0.25) -- (0.2,-0.25) -- (0.4,-0.25);\draw[black, thick] (0,-0.7) -- (0,0.7); \node at (-0.4,0) {$v_0$}; \node at (0.4,0) {$v_0$}; \node at (0.2,0.6) {$I$};} + \tikz[baseline=-3pt]{\draw[black, thick] (-0.4,-0.25) -- (-0.2,-0.25) -- (-0.2,0.25) -- (-0.4,0.25); \draw[black, thick] (0.4,0.25) -- (0.2,0.25) -- (0.2,-0.25) -- (0.4,-0.25);\draw[black, thick] (0,-0.7) -- (0,0.7);\node at (-0.4,0) {$v_1$}; \node at (0.4,0) {$v_1$}; \node at (0.2,0.6) {$X$}; } + \tikz[baseline=-3pt]{\draw[black, thick] (-0.4,-0.25) -- (-0.2,-0.25) -- (-0.2,0.25) -- (-0.4,0.25); \draw[black, thick] (0.4,0.25) -- (0.2,0.25) -- (0.2,-0.25) -- (0.4,-0.25);\draw[black, thick] (0,-0.7) -- (0,0.7); \node at (-0.4,0) {$v_3$}; \node at (0.4,0) {$v_0$}; \node at (0.2,0.6) {$Z$};}  + \tikz[baseline=-3pt]{\draw[black, thick] (-0.4,-0.25) -- (-0.2,-0.25) -- (-0.2,0.25) -- (-0.4,0.25); \draw[black, thick] (0.4,0.25) -- (0.2,0.25) -- (0.2,-0.25) -- (0.4,-0.25);\draw[black, thick] (0,-0.7) -- (0,0.7); \node at (-0.4,0) {$v_2$}; \node at (0.4,0) {$v_1$}; \node at (0.25,0.6) {$iY$}; } 
\end{equation}
\begin{equation}
\mathbb{M}(X) = \tikz[baseline=-3pt]{\draw[black, thick] (-0.5,-0.35) -- (0.5,-0.35); \draw[black, thick] (-0.5,0.35) -- (0.5,0.35);\draw[black, thick] (0,-0.9) -- (0,0.9); \node at (0.2,-0.5){$M$};\node at (0.25,0.55){$M^{\dagger}$};\filldraw[black] (0,0) circle (2pt) node[anchor=west]{X};} = \tikz[baseline=-3pt]{\draw[black, thick] (-0.4,-0.25) -- (-0.2,-0.25) -- (-0.2,0.25) -- (-0.4,0.25); \draw[black, thick] (0.4,0.25) -- (0.2,0.25) -- (0.2,-0.25) -- (0.4,-0.25);\draw[black, thick] (0,-0.7) -- (0,0.7); \node at (-0.4,0) {$v_3$}; \node at (0.4,0) {$v_3$}; \node at (0.2,0.6) {$I$}; } + \tikz[baseline=-3pt]{\draw[black, thick] (-0.4,-0.25) -- (-0.2,-0.25) -- (-0.2,0.25) -- (-0.4,0.25); \draw[black, thick] (0.4,0.25) -- (0.2,0.25) -- (0.2,-0.25) -- (0.4,-0.25);\draw[black, thick] (0,-0.7) -- (0,0.7); \node at (-0.4,0) {$v_0$}; \node at (0.4,0) {$v_3$}; \node at (0.2,0.6) {$Z$}; } + \tikz[baseline=-3pt]{\draw[black, thick] (-0.4,-0.25) -- (-0.2,-0.25) -- (-0.2,0.25) -- (-0.4,0.25); \draw[black, thick] (0.4,0.25) -- (0.2,0.25) -- (0.2,-0.25) -- (0.4,-0.25);\draw[black, thick] (0,-0.7) -- (0,0.7); \node at (-0.4,0) {$v_2$}; \node at (0.4,0) {$v_2$}; \node at (0.2,0.6) {$X$}; }  + \tikz[baseline=-3pt]{\draw[black, thick] (-0.4,-0.25) -- (-0.2,-0.25) -- (-0.2,0.25) -- (-0.4,0.25); \draw[black, thick] (0.4,0.25) -- (0.2,0.25) -- (0.2,-0.25) -- (0.4,-0.25);\draw[black, thick] (0,-0.7) -- (0,0.7); \node at (-0.4,0) {$v_1$}; \node at (0.4,0) {$v_2$}; \node at (0.25,0.6) {$iY$}; } 
\end{equation}
\begin{equation}
\mathbb{M}(Z) = \tikz[baseline=-3pt]{\draw[black, thick] (-0.5,-0.35) -- (0.5,-0.35); \draw[black, thick] (-0.5,0.35) -- (0.5,0.35);\draw[black, thick] (0,-0.9) -- (0,0.9); \node at (0.2,-0.5){$M$};\node at (0.25,0.55){$M^{\dagger}$};\filldraw[black] (0,0) circle (2pt) node[anchor=west]{Z};} = \tikz[baseline=-3pt]{\draw[black, thick] (-0.4,-0.25) -- (-0.2,-0.25) -- (-0.2,0.25) -- (-0.4,0.25); \draw[black, thick] (0.4,0.25) -- (0.2,0.25) -- (0.2,-0.25) -- (0.4,-0.25);\draw[black, thick] (0,-0.7) -- (0,0.7); \node at (-0.4,0) {$v_0$}; \node at (0.4,0) {$v_1$}; \node at (0.2,0.6) {$I$}; } + \tikz[baseline=-3pt]{\draw[black, thick] (-0.4,-0.25) -- (-0.2,-0.25) -- (-0.2,0.25) -- (-0.4,0.25); \draw[black, thick] (0.4,0.25) -- (0.2,0.25) -- (0.2,-0.25) -- (0.4,-0.25);\draw[black, thick] (0,-0.7) -- (0,0.7); \node at (-0.4,0) {$v_3$}; \node at (0.4,0) {$v_1$}; \node at (0.2,0.6) {$Z$}; } + \tikz[baseline=-3pt]{\draw[black, thick] (-0.4,-0.25) -- (-0.2,-0.25) -- (-0.2,0.25) -- (-0.4,0.25); \draw[black, thick] (0.4,0.25) -- (0.2,0.25) -- (0.2,-0.25) -- (0.4,-0.25);\draw[black, thick] (0,-0.7) -- (0,0.7); \node at (-0.4,0) {$v_1$}; \node at (0.4,0) {$v_0$}; \node at (0.2,0.6) {$X$}; }  + \tikz[baseline=-3pt]{\draw[black, thick] (-0.4,-0.25) -- (-0.2,-0.25) -- (-0.2,0.25) -- (-0.4,0.25); \draw[black, thick] (0.4,0.25) -- (0.2,0.25) -- (0.2,-0.25) -- (0.4,-0.25);\draw[black, thick] (0,-0.7) -- (0,0.7); \node at (-0.4,0) {$v_2$}; \node at (0.4,0) {$v_0$}; \node at (0.25,0.6) {$iY$}; } 
\end{equation}
where $v_0$ through $v_3$ are the four two-qubit Bell states. In particular, $v_i$ corresponds to an appropriately normalized doubled state of the Pauli matrix $\sigma^i$:
\begin{equation}
  \begin{aligned}
    |v_0\rangle = | I \rangle \rangle  =  \frac{1}{\sqrt{2}}|00\rangle + \frac{1}{\sqrt{2}}|11\rangle,\\
    |v_1\rangle= | X \rangle \rangle  = \frac{1}{\sqrt{2}}|01\rangle + \frac{1}{\sqrt{2}}|10\rangle,\\
    |v_2\rangle  = i|Y \rangle \rangle  = \frac{1}{\sqrt{2}}|01\rangle - \frac{1}{\sqrt{2}}|10\rangle,\\
    |v_3\rangle = | Z \rangle \rangle  = \frac{1}{\sqrt{2}}|00\rangle - \frac{1}{\sqrt{2}}|11\rangle
\end{aligned}  
\end{equation}

Connecting these tensors, we find
\begin{align}
\mathcal{D}^{\dagger} \times \mathcal{D} & =  \tikz[baseline=-3pt]{\draw[black, thick] (-0.5,-0.25) .. controls (-0.7,-0.25) and (-0.7,-0.35) .. (-0.6,-0.45);\draw[black, thick] (-0.5,-0.25) -- (0.5,-0.25); \draw[black, thick] (-0.5,0.25) .. controls (-0.7,0.25) and (-0.7,0.35) .. (-0.6,0.45);\draw[black, thick] (-0.5,0.25) -- (0.5,0.25);\draw[black, thick] (0,-0.7) -- (0,0.7); \node at (0.2,-0.4){$M$};\node at (0.25,0.45){$M^{\dagger}$};}\tikz[baseline=-3pt]{\draw[black, thick] (-0.5,-0.25) -- (0.5,-0.25); \draw[black, thick] (-0.5,0.25) -- (0.5,0.25);\draw[black, thick] (0,-0.7) -- (0,0.7);}\cdots \cdots \tikz[baseline=-3pt]{\draw[black, thick] (-0.5,-0.25) -- (0.5,-0.25); \draw[black, thick] (-0.5,0.25) -- (0.5,0.25);\draw[black, thick] (0,-0.7) -- (0,0.7); \draw[black, thick] (0.5,-0.25) .. controls (0.7,-0.25) and (0.7,-0.35) .. (0.6,-0.45);
\draw[black, thick] (0.5,0.25) .. controls (0.7,0.25) and (0.7,0.35) .. (0.6,0.45);} \\
 & =  I \otimes I ... \otimes I + X \otimes X ... \otimes X \\
 & =  I + \eta \label{eq:DdD}
\end{align}

\begin{align}
\mathcal{D}^{\dagger} X_i \mathcal{D} & =  \tikz[baseline=-3pt]{\draw[black, thick] (-0.5,-0.35) .. controls (-0.7,-0.35) and (-0.7,-0.45) .. (-0.6,-0.55);\draw[black, thick] (-0.5,-0.35) -- (0.5,-0.35); \draw[black, thick] (-0.5,0.35) .. controls (-0.7,0.35) and (-0.7,0.45) .. (-0.6,0.55);\draw[black, thick] (-0.5,0.35) -- (0.5,0.35);\draw[black, thick] (0,-0.7) -- (0,0.7); \node at (0.2,-0.5){$M$};\node at (0.25,0.55){$M^{\dagger}$};} \cdots \tikz[baseline=-3pt]{\draw[black, thick] (-0.5,-0.35) -- (0.5,-0.35); \draw[black, thick] (-0.5,0.35) -- (0.5,0.35);\draw[black, thick] (0,-0.9) -- (0,0.9); \node at (0.2,-0.5){$M$};\node at (0.25,0.55){$M^{\dagger}$};\filldraw[black] (0,0) circle (2pt) node[anchor=west]{X};}   \cdots \tikz[baseline=-3pt]{\draw[black, thick] (-0.5,-0.35) -- (0.5,-0.35); \draw[black, thick] (-0.5,0.35) -- (0.5,0.35);\draw[black, thick] (0,-0.7) -- (0,0.7); \draw[black, thick] (0.5,-0.35) .. controls (0.7,-0.35) and (0.7,-0.45) .. (0.6,-0.55);
\draw[black, thick] (0.5,0.35) .. controls (0.7,0.35) and (0.7,0.45) .. (0.6,0.55);} \\
 & =  Z_iZ_{i+1}\left(I+\eta\right) \label{eq:DdXD}
\end{align}

\begin{align}
\mathcal{D}^{\dagger} Z_iZ_{i+1} \mathcal{D} & =  \tikz[baseline=-3pt]{\draw[black, thick] (-0.5,-0.35) .. controls (-0.7,-0.35) and (-0.7,-0.45) .. (-0.6,-0.55);\draw[black, thick] (-0.5,-0.35) -- (0.5,-0.35); \draw[black, thick] (-0.5,0.35) .. controls (-0.7,0.35) and (-0.7,0.45) .. (-0.6,0.55);\draw[black, thick] (-0.5,0.35) -- (0.5,0.35);\draw[black, thick] (0,-0.7) -- (0,0.7); \node at (0.2,-0.5){$M$};\node at (0.25,0.55){$M^{\dagger}$};} \cdots \tikz[baseline=-3pt]{\draw[black, thick] (-0.5,-0.35) -- (0.5,-0.35); \draw[black, thick] (-0.5,0.35) -- (0.5,0.35);\draw[black, thick] (0,-0.9) -- (0,0.9); \node at (0.2,-0.5){$M$};\node at (0.25,0.55){$M^{\dagger}$};\filldraw[black] (0,0) circle (2pt) node[anchor=west]{Z};} 
\tikz[baseline=-3pt]{\draw[black, thick] (-0.5,-0.35) -- (0.5,-0.35); \draw[black, thick] (-0.5,0.35) -- (0.5,0.35);\draw[black, thick] (0,-0.9) -- (0,0.9); \node at (0.2,-0.5){$M$};\node at (0.25,0.55){$M^{\dagger}$};\filldraw[black] (0,0) circle (2pt) node[anchor=west]{Z};} \cdots \tikz[baseline=-3pt]{\draw[black, thick] (-0.5,-0.35) -- (0.5,-0.35); \draw[black, thick] (-0.5,0.35) -- (0.5,0.35);\draw[black, thick] (0,-0.7) -- (0,0.7); \draw[black, thick] (0.5,-0.35) .. controls (0.7,-0.35) and (0.7,-0.45) .. (0.6,-0.55);
\draw[black, thick] (0.5,0.35) .. controls (0.7,0.35) and (0.7,0.45) .. (0.6,0.55);} \\
 & =  X_{i+1}\left(I+\eta\right) \label{eq:DdZZD}
\end{align}
as well as
\begin{align}
\mathcal{D}^{\dagger} Z_i \mathcal{D} = 0
\end{align}
Due to the translation invariance of the matrix product representation of $\mathcal{D}$, these equations hold for all $i$. Eqs.~\eqref{eq:DdD}, \eqref{eq:DdXD}, \eqref{eq:DdZZD} exactly reproduce the fusion and the mapping relations of the Kramers-Wannier transformation. Therefore, starting from the sequential circuit that moves the symmetry twist, we have recovered the full symmetry action of Kramers-Wannier using the matrix product operator representation.

The full symmetry action can be implemented as the sequential circuit supplemented with nonunitary operations at the end. Given the decomposition of the unitary in the sequential circuit, Eq.~\eqref{eq:decomposition}, we can implement the full symmetry action as follows:
\begin{enumerate}
\item Introduce two ancilla qubits in the maximally entangled state $|v_0\rangle = \frac{1}{\sqrt{2}}|00\rangle + \frac{1}{\sqrt{2}}|11\rangle$. Place one near the first physical qubit and one near the last. 
\item Apply an operation controlled by the first ancilla qubit on the first physical qubit $|0\rangle \langle 0|\otimes I + |1\rangle \langle 1|\otimes Z$.
\item Apply the sequential unitaries $U_{\text{KW}} = u_{N-1,N}...u_{2,3}u_{1,2}$.
\item Apply an operation controlled by the second ancilla qubit on the last physical qubit $|0\rangle \langle 0|\otimes R(X) + |1\rangle \langle 1|\otimes iZR(X)$.
\item Project the two ancilla qubits back into the maximally entangled state $|v_0\rangle  = \frac{1}{\sqrt{2}}|00\rangle + \frac{1}{\sqrt{2}}|11\rangle$.
\end{enumerate}
By following the evolution of the state through this procedure, we see how full symmetry action is achieved. In fact, when using the ancilla to simulate the contraction of the virtual indices, we miss by an overall factor of $\sqrt{2}$. To correct for this, we modify the projection in the last step to be in the state $|00\rangle + |11\rangle$. 
\begin{equation}
\begin{array}{lll}
& & \frac{1}{\sqrt{2}}|\psi\rangle \left(|00\rangle + |11\rangle\right) \\
& \to & \frac{1}{\sqrt{2}}|\psi\rangle|00\rangle + \frac{1}{\sqrt{2}}Z_1|\psi\rangle|11\rangle \\
& \to & \frac{1}{\sqrt{2}}U_{\text{KW}}|\psi\rangle|00\rangle + \frac{1}{\sqrt{2}}U_{\text{KW}}Z_1|\psi\rangle|11\rangle \\
& \to & \frac{1}{\sqrt{2}}R(X_N)U_{\text{KW}}|\psi\rangle|00\rangle + i\frac{1}{\sqrt{2}}Z_NR(X_N)U_{\text{KW}}Z_1|\psi\rangle|11\rangle \\
& \to & \frac{1}{\sqrt{2}}\left(I + iZ_N\tilde{Z}_1\right)U_{\text{KW}}R(X_1) |\psi\rangle
\end{array}
\end{equation}
where $\tilde{Z}_1 = R(X_N)U_{\text{KW}}Z_1 U^{\dagger}_{\text{KW}}R^{\dagger}(X_N) = -iZ_N\eta$. The full symmetry action hence becomes
\begin{equation}
\mathcal{D} = \frac{(1+\eta)}{\sqrt{2}}U_{\text{KW}}R(X_1)
\end{equation}
which exactly matches the expression given in Ref.~\onlinecite{Shao2024TASI} up to a normalization factor of $\sqrt{2}$. This normalization is important to match the fusion rule $\mathcal{D}^{\dagger}\times\mathcal{D} = I + \eta$, where the coefficient before each symmetry is a positive integer.

\section{1D Generalized Symmetry as Matrix Product Operators}
\label{sec:1D}

The Kramers-Wannier example discussed in the last section captures all the key features of generalized symmetries in 1D. The bulk of their action is composed of a (unitary) sequential circuit. The full symmetry action is not necessarily unitary. Starting from the sequential circuit in the bulk, the full symmetry action can be obtained by supplementing the unitary circuit with possibly non-unitary transformations at the ends. The matrix product operator representation of the full symmetry operation can be obtained from the matrix product representation of the middle part of the sequential circuit by connecting copies of it in a translation-invariant form. 

We will call the matrix product operators obtained in this way a \textit{Sequential Matrix Product Operator} (sMPO for short). Note that not all matrix product operators can be obtained in this way from a sequential circuit. For example, the tensor product of projection operators
\begin{equation}
|0\rangle \langle 0| \otimes |0\rangle \langle 0| \otimes ... \otimes |0\rangle \langle 0|  
\end{equation}
can be written as a matrix product operator with virtual indices of dimension one (it is a tensor product operator), but it cannot be implemented as a sequential circuit with non-unitaries at the endpoints. The (translation-invariant) matrix product unitary operators studied in Ref.~\onlinecite{Cirac2017,Sahinoglu2018} is a subset of sMPO, which represent invertible and locality preserving symmetries. The non-translation invariant matrix product unitary operators studied in Ref.~\onlinecite{Styliaris2025} contains sMPO as a subset.

The sMPO representation of generalized symmetries in 1D reveals many of their important properties, which we demonstrate in this section:
\begin{enumerate}
\item 1D generalized symmetries form a closed algebra under composition (fusion) and linear superposition. `Simple' objects in the space of symmetry operators are short-range correlated operators.
\item The fusion of two simple symmetries can result in the sum of more than one simple symmetry, each with a non-negative integer coefficient.
\item The result of fusion of a simple symmetry with its Hermitian conjugate contains one and only one summand that is the identity.
\item The fusion of the Hermitian conjugate of a symmetry with a different symmetry cannot contain the identity channel.
\end{enumerate}

This shows essentially that 1D generalized symmetries represented by sMPOs form a semisimple tensor category\footnote{Unlike a fusion category, the number of simple objects do not need to be finite. Thus, sMPOs can also capture continuous non-invertible symmetries}. The simple symmetries are the simple objects in the category. The identity operator is the tensor unit in the category. The Hermitian conjugate of a simple symmetry is the dual of the object and the dual is unique. The morphism space $V^{c}_{ab}=\text{Hom}(a\otimes b, c)$ is a vector space with non-negative integer dimension $N^c_{ab}$.

Our proofs in this section apply only to generalized symmetries acting on a tensor product of finite-dimensional Hilbert spaces (like the Kramers-Wannier operator). But all the results can be generalized to symmetries acting on constrained non-tensor-product Hilbert spaces as well, which include interesting examples like the Fibonacci symmetry~\cite{Feiguin2007}. Appendix~\ref{app:constrained} contains the proofs of claims~\ref{claim1} through \ref{claim4} for the general case.

\begin{claim}
%\label{claim:algebra}
1D generalized symmetries, represented by Sequential Matrix Product Operators, form a closed algebra under composition (fusion) and linear superposition.
\label{claim1}
\end{claim}

Suppose that we have two generalized symmetries $\mathcal{D}_{\alpha}$ and $\mathcal{D}_{\beta}$ and each can be implemented using the procedure illustrated in the last section for the Kramer-Wannier transformation. We start with a maximally entangled pair of ancillas corresponding to the virtual indices of the MPO representation, apply a sequential circuit to the ancillas and the physical degrees of freedom, and finally decouple the ancillas from the physical DOF using certain projection operation. The goal is to show that the composition of these two symmetries $\mathcal{D}_{\beta}\times\mathcal{D}_{\alpha}$ as well as any linear combination of these two symmetries $a\mathcal{D}_{\alpha}+b\mathcal{D}_{\beta}$ can be implemented using the same procedure.

\begin{proof}
Suppose that for the implementation of $\mathcal{D}_{\alpha}$ we need ancillas with dimension $d_{\alpha}$ and for $\mathcal{D}_{\beta}$ we need $d_{\beta}$. 

To implement $\mathcal{D}_{\beta}\times\mathcal{D}_{\alpha}$, we start with two pairs of ancillas with dimension $d_{\alpha}$ and $d_{\beta}$ respectively, initialized in their maximally entangled state $\frac{1}{\sqrt{d_{\alpha}}}\sum_{i=1}^{d_{\alpha}} |ii\rangle$ and $\frac{1}{\sqrt{d_{\beta}}}\sum_{j=1}^{d_{\beta}} |jj\rangle$. Then we apply the sequential circuit part of $\mathcal{D}_{\alpha}$ and $\mathcal{D}_{\beta}$,
\begin{equation}
\begin{array}{l}
U_{\alpha} = u^{\alpha}_{N-1,N}u^{\alpha}_{N-2,N-1}...u^{\alpha}_{1,2} \\
U_{\beta} = u^{\beta}_{N-1,N}u^{\beta}_{N-2,N-1}...u^{\beta}_{1,2}
\end{array}
\end{equation}
in a parallel way such that they combine into a single sequential circuit. For example, we can apply the unitaries in the two circuits following this order
\begin{equation}
u^{\beta}_{N-1,N}u^{\beta}_{N-2,N-1}u^{\alpha}_{N-1,N}...u^{\alpha}_{4,5}u^{\beta}_{2,3}u^{\alpha}_{3,4}u^{\beta}_{1,2}u^{\alpha}_{2,3}u^{\alpha}_{1,2}
\end{equation}
which is equivalent to $U_{\beta}U_{\alpha}$
but organized in a sequential way. Finally, we project the two pairs of ancillas back into their maximally entangled state and complete the implementation of $\mathcal{D}_{\beta}\times\mathcal{D}_{\alpha}$.

To implement $a\mathcal{D}_{\alpha}+b\mathcal{D}_{\beta}$, we take a pair of ancillas with dimension $d_{\alpha}+d_{\beta}$ and initialize them in the entangled state $\frac{a}{\sqrt{d_{\alpha}}}\sum_{i=1}^{d_{\alpha}} |ii\rangle + \frac{b}{\sqrt{d_{\beta}}}\sum_{i=d_{\alpha}+1}^{d_{\alpha}+d_{\beta}} |ii\rangle$. Then we apply a controlled sequential circuit using the ancilla as a control: if the ancilla is in the $d_{\alpha}$ dimensional subspace, we apply the $u^{\alpha}$ gates and if the ancilla is in the $d_{\beta}$ dimensional subspace, we apply the $u^{\beta}$ gates. This can be done in a sequential way. Finally, we project the ancillas in the maximally entangled state $\frac{1}{\sqrt{d_{\alpha}}}\sum_{i=1}^{d_{\alpha}} |ii\rangle + \frac{1}{\sqrt{d_{\beta}}}\sum_{i=d_{\alpha}+1}^{d_{\alpha}+d_{\beta}} |ii\rangle$and complete the implementation of $a\mathcal{D}_{\alpha}+b\mathcal{D}_{\beta}$.
\label{proof:1}
\end{proof}

Therefore, generalized symmetries in a 1D system form an algebra -- a vector space equipped with a product. To match the normalization of the translation invariant matrix product operator obtained from the sequential circuit, the normalization of the non-unitary step at the end needs to be properly chosen. We will discuss the normalization more carefully below. Among all the operators in the vector space, some are special and called `simple'. The simple symmetries form a basis to decompose any generalized symmetry operator in the vector space. Matrix product operators, on the other hand, has a natural decomposition into `injective' matrix product operators given by its canonical form as reviewed in Appendix~\ref{app:MPS}. Therefore, we expect `simple' symmetries to be represented by `injective' matrix product operators. Injectivity means that the matrix product operator is short range correlated. That is, a simple object $\mathcal{D}_{\alpha}$ in a fusion category symmetry satisfies
\begin{equation}
\begin{array}{ll}
\frac{1}{\mathcal{N}}\text{Tr}\left(O_iO_j\mathcal{D}_{\alpha}^{\dagger}O'_iO'_j\mathcal{D}_{\alpha}\right) - \\ \frac{1}{\mathcal{N}^2}\text{Tr}\left(O_i\mathcal{D}_{\alpha}^{\dagger}O'_i\mathcal{D}_{\alpha}\right)\text{Tr}\left(O_j\mathcal{D}_{\alpha}^{\dagger}O'_j\mathcal{D}_{\alpha}\right) \\
\sim \exp(-|i-j|/\xi_{\alpha})
\end{array}
\end{equation}
when the distance between $i$ and $j$, $|i-j|$, becomes large. Here, $\mathcal{N}$ is the normalization $\text{Tr}\left(\mathcal{D}^{\dagger}_{\alpha}\mathcal{D}_{\alpha}\right)$. This condition holds for any local operators $O_i$, $O_j$, $O'_i$, $O'_j$. An equivalent way of saying this is that if the input and output physical indices at each site are combined into a single physical index, the 1D operator becomes a short range correlated 1D state with correlation length $\xi_{\alpha}$. 
\begin{align}
\tikz[baseline=0pt]{\draw[black, thick] (-3,0) .. controls (-3.2,0) and (-3.2,0.10) .. (-3.1,0.20);\draw[black, thick] (-3,0) -- (3,0); \draw[black, thick] (0,-0.5) -- (0,0.5); \draw[black, thick] (-2.4,-0.5) -- (-2.4,0.5);\draw[black, thick] (-1.6,-0.5) -- (-1.6,0.5);\draw[black, thick] (-0.8,-0.5) -- (-0.8,0.5);\draw[black, thick] (2.4,-0.5) -- (2.4,0.5);\draw[black, thick] (1.6,-0.5) -- (1.6,0.5);\draw[black, thick] (0.8,-0.5) -- (0.8,0.5);\node at (0.2,-0.2){$M_{\alpha}$};\draw[black, thick] (3,0) .. controls (3.2,0) and (3.2,0.10) .. (3.1,0.20);} \\ \nonumber
\to \tikz[baseline=0pt]{\draw[black, thick] (-3,0) .. controls (-3.2,0) and (-3.2,0.10) .. (-3.1,0.20);\draw[black, thick] (-3,0) -- (3,0); \draw[black, thick] (0,0.5) -- (0,0) .. controls (0.1,-0.1).. (0.2,0)--(0.2,0.45); \draw[black, thick] (-2.4,0.5) -- (-2.4,0) .. controls (-2.3,-0.1).. (-2.2,0)--(-2.2,0.45); \draw[black, thick] (-1.6,0.5) -- (-1.6,0) .. controls (-1.5,-0.1).. (-1.4,0)--(-1.4,0.45);\draw[black, thick] (-0.8,0.5) -- (-0.8,0) .. controls (-0.7,-0.1).. (-0.6,0)--(-0.6,0.45);\draw[black, thick] (2.4,0.5) -- (2.4,0) .. controls (2.5,-0.1).. (2.6,0)--(2.6,0.45);\draw[black, thick] (1.6,0.5) -- (1.6,0) .. controls (1.7,-0.1).. (1.8,0)--(1.8,0.45);\draw[black, thick] (0.8,0.5) -- (0.8,0) .. controls (0.9,-0.1).. (1.0,0)--(1.0,0.45);\node at (0.2,-0.2){$A_{\alpha}$};\draw[black, thick] (3,0) .. controls (3.2,0) and (3.2,0.10) .. (3.1,0.20);} 
\label{eq:M2A}
\end{align}
The matrix product state representation of the 1D state satisfies the `injectivity' condition that its canonical form contains only one diagonal block. Moreover, the transfer matrix 
\begin{align}
\mathbb{T}_{\alpha} = \tikz[baseline=-3pt]{\draw[black, thick] (-0.5,-0.25) -- (0.5,-0.25); \draw[black, thick] (-0.5,0.25) -- (0.5,0.25);\draw[black, thick] (0,-0.5) -- (0,0.5); \draw[black, thick] (-0.2,-0.5) -- (-0.2,0.5);\draw[black,thick] (0,-0.5)..controls (-0.1,-0.6)..(-0.2,-0.5);\draw[black,thick] (0,0.5)..controls (-0.1,0.6)..(-0.2,0.5);\node at (0.3,-0.45){$A_{\alpha}$}; \node at (0.3,0.45){$A^{\dagger}_{\alpha}$};}
\end{align}
has a nondegenerate largest eigenvalue ($\lambda_{\alpha}>0$) and the corresponding eigenvector takes the form of $\sum_i |ii\rangle$ on the two virtual indices (represented by the square bracket in the equation below). 
\begin{align}
\mathbb{T}_{\alpha} \left(\sum_{i}|ii\rangle\right) = \tikz[baseline=-3pt]{\draw[black, thick] (-0.5,-0.25) -- (0.5,-0.25); \draw[black, thick] (-0.5,0.25) -- (0.5,0.25);\draw[black, thick] (0,-0.5) -- (0,0.5); \draw[black, thick] (-0.2,-0.5) -- (-0.2,0.5);\draw[black,thick] (0,-0.5)..controls (-0.1,-0.6)..(-0.2,-0.5);\draw[black,thick] (0,0.5)..controls (-0.1,0.6)..(-0.2,0.5);\draw[black,thick] (0.55,0.25) -- (0.65,0.25)--(0.65,-0.25)--(0.55,-0.25)} = \lambda_{\alpha}\tikz[baseline=-3pt]{\draw[black,thick] (0.02,0.25) -- (0.12,0.25)--(0.12,-0.25)--(0.02,-0.25)}
\label{eq:dt}
\end{align}
Note that here, we do not rescale $\lambda_{\alpha}$ to $1$ as is usually done for matrix product states because the overall normalization of the MPO matters. We will see that because the MPO comes from a sequential circuit, $\lambda_{\alpha} = d$, where $d$ is the dimension of the physical index.

On-site symmetries of the tensor product form are obviously short range correlated with correlation length $\xi_{\alpha} = 0$. The Kramers-Wannier operator $\mathcal{D}$ is also short-range correlated with correlation length $\xi_{\alpha}=0$. In fact, the corresponding short range correlated 1D state is just the 1D cluster state.

\begin{claim}
The fusion of two simple symmetries can result in the sum of more than one simple symmetry, each with a non-negative integer coefficient.
\label{claim2}
\end{claim}

The fusion of two `simple' symmetry operators may no longer be simple.  The Kramers-Wannier transformation is a typical example. The fusion rule $\mathcal{D}^{\dagger} \times \mathcal{D} = 1+\eta$ is manifested in the MPO representation as
\begin{equation}
\mathbb{M} = \tikz[baseline=-3pt]{\draw[black, thick] (-0.5,-0.25) -- (0.5,-0.25); \draw[black, thick] (-0.5,0.25) -- (0.5,0.25);\draw[black, thick] (0,-0.7) -- (0,0.7); \node at (0.2,-0.4){$M$};\node at (0.25,0.45){$M^{\dagger}$};} = \tikz[baseline=-3pt]{\draw[black, thick] (-0.4,-0.25) -- (-0.2,-0.25) -- (-0.2,0.25) -- (-0.4,0.25); \draw[black, thick] (0.4,0.25) -- (0.2,0.25) -- (0.2,-0.25) -- (0.4,-0.25);\draw[black, thick] (0,-0.7) -- (0,0.7); \node at (-0.4,0) {$v_0$}; \node at (0.4,0) {$v_0$}; \node at (0.2,0.6) {$I$};} + \tikz[baseline=-3pt]{\draw[black, thick] (-0.4,-0.25) -- (-0.2,-0.25) -- (-0.2,0.25) -- (-0.4,0.25); \draw[black, thick] (0.4,0.25) -- (0.2,0.25) -- (0.2,-0.25) -- (0.4,-0.25);\draw[black, thick] (0,-0.7) -- (0,0.7);\node at (-0.4,0) {$v_1$}; \node at (0.4,0) {$v_1$}; \node at (0.2,0.6) {$X$}; } + ...
\end{equation}
where only the diagonal blocks are shown and all the off-diagonal blocks in the decomposition are omitted (...). We see that the fused MPO is not injective, but rather contains two diagonal blocks, corresponding to the $I$ and $\eta$ operators in the fusion result. In general, we have
\begin{equation}
\mathbb{M}_{\alpha,\beta} = \tikz[baseline=-3pt]{\draw[black, thick] (-0.5,-0.25) -- (0.5,-0.25); \draw[black, thick] (-0.5,0.25) -- (0.5,0.25);\draw[black, thick] (0,-0.7) -- (0,0.7); \node at (0.3,-0.45){$M_{\alpha}$};\node at (0.3,0.45){$M_{\beta}$};} = \sum_{\gamma} \tikz[baseline=-3pt]{\draw[black, thick] (-0.5,-0.25) -- (0.5,-0.25); \draw[black, thick] (-0.5,0.25) -- (0.5,0.25);\draw[black, thick] (0,-0.7) -- (0,0.7); \node at (0.3,0){$M_{\gamma}$};} + ...
\label{eq:FusionMPO}
\end{equation}
where the right hand side of the second $=$ sign gives the canonical decomposition of the tensor $\mathbb{M}_{\alpha,\beta}$. Only the diagonal blocks are shown and all the off-diagonal blocks in the decomposition are omitted (...). The properties of the canonical decomposition tell us that the $M_{\gamma}$ tensors are each supported on orthogonal spaces $V_{\gamma}$ of the virtual indices $V_{\alpha} \otimes V_\beta$, and their represented operators $\mathcal{D}_{\gamma}$ are short-range correlated.

We will show, based on this decomposition, that each $M_{\gamma}$ represents a sequential matrix product operator. That is, the corresponding 1D operator $\mathcal{D}_{\gamma}$ can be implemented using a sequential circuit (with ancillas) plus non-unitary operations at the end. Therefore, each $\mathcal{D}_{\gamma}$ is a generalized symmetry and the fusion of two simple symmetries can result in a sum of simple symmetries, $\mathcal{D}_{\beta} \times \mathcal{D}_{\alpha} = \sum_{\gamma} \mathcal{D}_{\gamma}$. Different $\gamma$ values can give rise to the same generalized symmetry operator. When that happens, the corresponding symmetry has a nontrivial positive integer multiplicity $N^{\gamma}_{\alpha\beta}$ in the fusion result. 
\begin{equation}
\mathcal{D}_{\beta} \times \mathcal{D}_{\alpha} = \sum_{\gamma} N_{\alpha\beta}^{\gamma} \mathcal{D}_{\gamma}, \ \ N^{\gamma}_{\alpha\beta} \in \mathbb{Z}_{\geq 0}
\end{equation}

\begin{proof}
To show that $\mathcal{D}_{\gamma}$ can be implemented using a sequential circuit (with ancillas) plus non-unitary operations at the end, we just need to slightly modify the final step in the implementation of $\mathcal{D}_{\beta}\times \mathcal{D}_{\alpha}$. We start again with two pairs of ancillas with dimension $d_{\alpha}$ and $d_{\beta}$ respectively and initialized in their maximally entangled state $\frac{1}{\sqrt{d_{\alpha}}}\sum_{i=1}^{d_{\alpha}} |ii\rangle$ and $\frac{1}{\sqrt{d_{\beta}}}\sum_{j=1}^{d_{\beta}} |jj\rangle$. Then, we apply the sequential circuit part of $\mathcal{D}_{\alpha}$ and $\mathcal{D}_{\beta}$ in a parallel way such that they combine into a single sequential circuit, as discussed in the proof of Claim \ref{claim1}. At the last step, we take the ancilla's for both $\alpha$ and $\beta$ and project them into the maximally entangled state in $V_{\gamma}$, the subspace of virtual indices on which $\lambda_{\gamma}M_{\gamma}$ is supported. Such a projection allows us to pick out the $\lambda_{\gamma}M_{\gamma}$ component in the composed operation of $M_{\alpha}$ and $M_{\beta}$. Therefore, each $\lambda_{\gamma}M_{\gamma}$ is an sMPO and each $\mathcal{D}_{\gamma}$ can be implemented with a sequential circuit plus non-unitary operations at the end.

The fact that each $M_{\gamma}$ is an sMPO fixes their normalization. It is not possible to add a prefactor $|\lambda_{\gamma}| \neq 1$ to $M_{\gamma}$ while maintaining the sMPO properly. This is because a prefactor $|\lambda_{\gamma}| \neq 1$ changes the norm of the resulting operator $\mathcal{D}_{\gamma}$ by $\lambda_{\gamma}^N$, where $N$ is the length of the 1d chain. On the other hand, $\mathcal{D}_{\gamma}$ can be implemented as a sequential circuit with non-unitary operations only at the end, indicating that the norm of $\mathcal{D}_{\gamma}$ can only differ from that of a unitary operator by a constant factor which does not scale with the system size. Therefore, the normalization of each $M_{\gamma}$ is fixed and there can be no prefactor $|\lambda_{\gamma}| \neq 1$. In the proof of claim~\ref{claim3}, we will see that the normalization of $M_{\gamma}$ is such that their spectral radius is $d$, the dimension of each physical index.

With the normalization of $\mathcal{D}_{\gamma}$ fixed by that given by $M_{\gamma}$, the prefactor in front of each $\mathcal{D}_{\gamma}$ in $\mathcal{D}_{\beta} \times \mathcal{D}_{\alpha} = \sum_{\gamma} \mathcal{D}_{\gamma}$ is fixed to be $1$ in this decomposition. Multiplicity $N^{\gamma}_{\alpha\beta}>1$ can only come due to different $\gamma$ blocks representing the same symmetry operator. 
\end{proof}

\begin{claim}
The result of the fusion of a simple symmetry with its Hermitian conjugate contains one and only one summand that is the identity.
\label{claim3}
\end{claim}

Suppose that the simple symmetry $\mathcal{D}_{\alpha}$ is represented by an injective MPO $M_{\alpha}$. The fusion of $\mathcal{D}_{\alpha}^{\dagger}$ and $\mathcal{D}_{\alpha}$ is represented by the tensor
\begin{equation}
\mathbb{M}_{\alpha} = \tikz[baseline=-3pt]{\draw[black, thick] (-0.5,-0.25) -- (0.5,-0.25); \draw[black, thick] (-0.5,0.25) -- (0.5,0.25);\draw[black, thick] (0,-0.7) -- (0,0.7); \node at (0.25,-0.45){$M_{\alpha}$};\node at (0.25,0.45){$M^{\dagger}_{\alpha}$};} 
\label{eq:Ma}
\end{equation}
To prove this claim, we will show that in the canonical decomposition of $\mathbb{M}_{\alpha}$, there is one and only one diagonal block with the physical indices in the $I$ matrix form
\begin{equation}
\mathbb{M}_{\alpha} = \tikz[baseline=-3pt]{\draw[black, thick] (-0.5,-0.25) -- (0.5,-0.25); \draw[black, thick] (-0.5,0.25) -- (0.5,0.25);\draw[black, thick] (0,-0.7) -- (0,0.7); \node at (0.2,-0.4){$M_{\alpha}$};\node at (0.25,0.45){$M^{\dagger}_{\alpha}$};} = \tikz[baseline=-3pt]{\draw[black, thick] (-0.4,-0.25) -- (-0.2,-0.25) -- (-0.2,0.25) -- (-0.4,0.25); \draw[black, thick] (0.4,0.25) -- (0.2,0.25) -- (0.2,-0.25) -- (0.4,-0.25);\draw[black, thick] (0,-0.7) -- (0,0.7); \node at (-0.4,0) {$v_0$}; \node at (0.4,0) {$v_0$}; \node at (0.2,0.6) {$I$};} +  ...
\label{eq:Ma1}
\end{equation}
... includes all the other diagonal blocks as well as off-diagonal blocks in the decomposition. We will call a block having the physical indices in the $I$ matrix form an `identity block' for short. 

A direct consequence of this claim is that, 1D non-invertible symmetries are always non-invertible in the weaker sense because when $\mathcal{D}_{\alpha}^{\dagger}$ is fused with $\mathcal{D}_{\alpha}$, one and only one of the fusion channels is identity $I$.
\begin{equation}
\mathcal{D}^{\dagger}_{\alpha} \times \mathcal{D}_{\alpha} = I + ...
\label{eq:1Dnoninvert}
\end{equation}
That is, they are annihilable.

\begin{proof}
First, we argue that there must be at least one identity block on the diagonal. The MPO representation of the $\mathcal{D}_{\alpha}$ symmetry comes from that of a sequential unitary circuit. When we impose a periodic boundary condition on $M_{\alpha}$, we get $\mathcal{D}_{\alpha}$. But with a properly chosen open boundary condition, we should be able to recover the unitary sequential circuit. Correspondingly, with a properly chosen open boundary condition, $\mathbb{M}_{\alpha}$ should represent the identity operator $I\otimes ...\otimes I$ on a chain of arbitrary length. This is possible only if at least one diagonal block in the canonical decomposition of $\mathbb{M}_{\alpha}$ is the identity block. If identity blocks only show up on the off-diagonal, since all the off-diagonal blocks are in the upper triangle region (see Eq.~\ref{eq:MPScano}) and there is a finite number of them, it is hence not possible to generate $I\otimes ...\otimes I$ on a chain of arbitrary length, even with open boundary conditions.

Consider the decomposition shown in Eq.~\ref{eq:Ma1} with one diagonal identity block. The diagonal identity block gives rise to a $I\otimes...\otimes I$ term in the fusion of $\mathcal{D}^{\dagger}_{\alpha}$ and $\mathcal{D}_{\alpha}$. Since $I\otimes...\otimes I$ is a tensor product operator, this diagonal block has a virtual index of dimension one, supported on the vector $v_0$. Note that, since we proved in Claim~\ref{claim2} that each diagonal component in this decomposition can be implemented with a sequential circuit (which is trivially true for $I\otimes...\otimes I$), the normalization of this term is fixed. $v_0$ is a normalized vector (with length $1$) and there can be no prefactor in front of this term. Now, let us contract the two physical indices of $\mathbb{M}_{\alpha}$ and obtain the transfer matrix $\mathbb{T}_{\alpha}$ of $A_{\alpha}$. We see that $v_0$ is an eigenvector of $\mathbb{T}_{\alpha}$ with eigenvalue $d$, the dimension of the physical index. 
\begin{align}
\mathbb{T}_{\alpha} |v_0\rangle = \tikz[baseline=-3pt]{\draw[black, thick] (-0.5,-0.25) -- (0.5,-0.25); \draw[black, thick] (-0.5,0.25) -- (0.5,0.25);\draw[black, thick] (0,-0.5) -- (0,0.5); \draw[black, thick] (-0.2,-0.5) -- (-0.2,0.5);\draw[black,thick] (0,-0.5)..controls (-0.1,-0.6)..(-0.2,-0.5);\draw[black,thick] (0,0.5)..controls (-0.1,0.6)..(-0.2,0.5);\draw[black,thick] (0.7,0.25) -- (0.8,0.25)--(0.8,-0.25)--(0.7,-0.25); \node at (0.25,-0.45){$A_{\alpha}$};\node at (0.25,0.45){$A^{\dagger}_{\alpha}$}; \node at (0.6,0) {$v_0$}; } = d\tikz[baseline=-3pt]{\draw[black,thick] (0.3,0.25) -- (0.4,0.25)--(0.4,-0.25)--(0.3,-0.25);\node at (0.2,0) {$v_0$};}
\label{eq:dt1}
\end{align}
Note that only the first term (the term shown in Eq.~\ref{eq:Ma1}) contributes to the action of $\mathbb{T}_{\alpha}$ on $|v_0\rangle$. The action of all the other terms (the ones not shown) on $|v_0\rangle$ gives zero due to the upper triangular structure of the canonical form (see Eq.~\ref{eq:MPScano}).

If there is a second identity block on the diagonal, we would have another eigenvector $|v_1\rangle$ of $\mathbb{T}_{\alpha}$ with eigenvalue $d$. However, this is not possible. We prove this by showing that $d$ is the largest (in magnitude) eigenvalue of $\mathbb{T}_{\alpha}$ and hence for an injective MPO $M_{\alpha}$, the corresponding eigenvector $|v_0\rangle$ is unique.

To see that $d$ is the largest eigenvalue of $\mathbb{T}_{\alpha}$, we note that since $\mathcal{D}_{\alpha}^{\dagger}\times \mathcal{D}_{\alpha}$ differs from $U_{\alpha}^{\dagger}U_{\alpha} = I\otimes ... \otimes I$ only by non-unitary operations at the boundary, the trace of the operator is bounded by that of $I\otimes ... \otimes I$ up to a constant factor, say $c_{\alpha}$
\begin{equation}
\text{Tr}\left(\mathcal{D}_{\alpha}^{\dagger}\times \mathcal{D}_{\alpha}\right) \leq c_{\alpha} \text{Tr}\left(I\otimes ...\otimes I\right)=c_{\alpha} d^N
\label{eq:traceDD}
\end{equation}
where $d$ is the dimension of each physical index. On the other hand, the trace of $\mathcal{D}_{\alpha}^{\dagger}\times \mathcal{D}_{\alpha}$ is equal to the contraction of copies of $\mathbb{T}_{\alpha}$ along the chain
\begin{align}
\tikz[baseline=0pt]{\draw[black, thick] (-0.5,-0.25) -- (0.5,-0.25); \draw[black, thick] (-0.5,0.25) -- (0.5,0.25);\draw[black, thick] (0,-0.5) -- (0,0.5); \draw[black, thick] (-0.2,-0.5) -- (-0.2,0.5);\draw[black,thick] (0,-0.5)..controls (-0.1,-0.6)..(-0.2,-0.5);\draw[black,thick] (0,0.5)..controls (-0.1,0.6)..(-0.2,0.5);\draw[black, thick] (-0.5,0.25) .. controls (-0.7,0.25) and (-0.7,0.35) .. (-0.6,0.45);\draw[black, thick] (-0.5,-0.25) .. controls (-0.7,-0.25) and (-0.7,-0.35) .. (-0.6,-0.45);}\tikz[baseline=0pt]{\draw[black, thick] (-0.5,-0.25) -- (0.5,-0.25); \draw[black, thick] (-0.5,0.25) -- (0.5,0.25);\draw[black, thick] (0,-0.5) -- (0,0.5); \draw[black, thick] (-0.2,-0.5) -- (-0.2,0.5);\draw[black,thick] (0,-0.5)..controls (-0.1,-0.6)..(-0.2,-0.5);\draw[black,thick] (0,0.5)..controls (-0.1,0.6)..(-0.2,0.5);\node at (0.2,0){$\mathbb{T}_{\alpha}$};}\tikz[baseline=0pt]{\draw[black, thick] (-0.5,-0.25) -- (0.5,-0.25); \draw[black, thick] (-0.5,0.25) -- (0.5,0.25);\draw[black, thick] (0,-0.5) -- (0,0.5); \draw[black, thick] (-0.2,-0.5) -- (-0.2,0.5);\draw[black,thick] (0,-0.5)..controls (-0.1,-0.6)..(-0.2,-0.5);\draw[black,thick] (0,0.5)..controls (-0.1,0.6)..(-0.2,0.5);}...\tikz[baseline=0pt]{\draw[black, thick] (-0.5,-0.25) -- (0.5,-0.25); \draw[black, thick] (-0.5,0.25) -- (0.5,0.25);\draw[black, thick] (0,-0.5) -- (0,0.5); \draw[black, thick] (-0.2,-0.5) -- (-0.2,0.5);\draw[black,thick] (0,-0.5)..controls (-0.1,-0.6)..(-0.2,-0.5);\draw[black,thick] (0,0.5)..controls (-0.1,0.6)..(-0.2,0.5);}\tikz[baseline=0pt]{\draw[black, thick] (-0.5,-0.25) -- (0.5,-0.25); \draw[black, thick] (-0.5,0.25) -- (0.5,0.25);\draw[black, thick] (0,-0.5) -- (0,0.5); \draw[black, thick] (-0.2,-0.5) -- (-0.2,0.5);\draw[black,thick] (0,-0.5)..controls (-0.1,-0.6)..(-0.2,-0.5);\draw[black,thick] (0,0.5)..controls (-0.1,0.6)..(-0.2,0.5);}\tikz[baseline=0pt]{\draw[black, thick] (-0.5,-0.25) -- (0.5,-0.25); \draw[black, thick] (-0.5,0.25) -- (0.5,0.25);\draw[black, thick] (0,-0.5) -- (0,0.5); \draw[black, thick] (-0.2,-0.5) -- (-0.2,0.5);\draw[black,thick] (0,-0.5)..controls (-0.1,-0.6)..(-0.2,-0.5);\draw[black,thick] (0,0.5)..controls (-0.1,0.6)..(-0.2,0.5);\draw[black, thick] (0.5,0.25) .. controls (0.7,0.25) and (0.7,0.35) .. (0.6,0.45);\draw[black, thick] (0.5,-0.25) .. controls (0.7,-0.25) and (0.7,-0.35) .. (0.6,-0.45);}
\label{eq:Ta}
\end{align}
Therefore, the largest eigenvalue of $\mathbb{T}_{\alpha}$ must be $d$. 
\end{proof}
From the above discussion, we can find a natural normalization of the injective sMPO $M_{\alpha}$: its spectral radius is equal to $d$.

\begin{claim}
The fusion of the Hermitian conjugate of a symmetry with a different symmetry cannot contain the identity channel.
\label{claim4}
\end{claim}
\begin{proof}
Suppose that two different simple symmetries $\mathcal{D}_{\alpha}$ and $\mathcal{D}_{\beta}$ are represented with injective tensors $M_{\alpha}$ and $M_{\beta}$ (and their associated MPS form $A_{\alpha}$ and $A_{\beta}$, see Eq.~\ref{eq:M2A}) respectively. $M_{\alpha}$ and $M_{\beta}$ are normalized with spectral radius $d$. Define the transfer matrix $\mathbb{T}_{\alpha,\beta^{\dagger}}$, which can be obtained from $\mathbb{M}_{\alpha,\beta^{\dagger}}$ by contracting the two physical indices. 
\begin{align}
\mathbb{T}_{\alpha,\beta^{\dagger}}  = \tikz[baseline=-3pt]{\draw[black, thick] (-0.5,-0.25) -- (0.5,-0.25); \draw[black, thick] (-0.5,0.25) -- (0.5,0.25);\draw[black, thick] (0,-0.5) -- (0,0.5); \draw[black, thick] (-0.2,-0.5) -- (-0.2,0.5);\draw[black,thick] (0,-0.5)..controls (-0.1,-0.6)..(-0.2,-0.5);\draw[black,thick] (0,0.5)..controls (-0.1,0.6)..(-0.2,0.5);\node at (0.25,-0.45){$A_{\alpha}$};\node at (0.25,0.45){$A^{\dagger}_{\beta}$}; }
\label{eq:dt12}
\end{align}
If $\mathcal{D}_{\beta}^{\dagger}\times \mathcal{D}_{\alpha} = I + ...$, then $\mathbb{M}_{\alpha,\beta^{\dagger}}$ contains an identity block on the diagonal. As a consequence, the transfer matrix $\mathbb{T}_{\alpha,\beta^{\dagger}}$ has an eigenvector with eigenvalue $d$. 

Now, to see the contradiction, consider the matrix product state $|\psi_{\alpha}\rangle$ and $|\psi_{\beta}\rangle$ represented by $A_{\alpha}$ and $A_{\beta}$. The norm of $|\psi_{\alpha}\rangle$ and is given by the contraction of copies of $\mathbb{T}_{\alpha}$ along the chain as shown in Eq.~\ref{eq:Ta}. Since $\mathbb{T}_{\alpha}$ has a largest nondegenerate eigenvalue of $d$, the norm of $|\psi_{\alpha}\rangle$ is $d^{N/2}$. The same is true for $|\psi_{\beta}\rangle$. The inner product between the two states $\langle \psi_{\beta} | \psi_{\alpha}\rangle$ is given by the contraction of copies of $\mathbb{T}_{\alpha,\beta^{\dagger}}$ along the chain (replace $\mathbb{T}_{\alpha}$ with $\mathbb{T}_{\alpha,\beta^{\dagger}}$ in Eq.~\ref{eq:Ta}). Since the inner product of two different matrix product states is always exponentially small than 1 when divided by their norm, the eigenvalue of the transfer matrix $\mathbb{T}_{\alpha,\beta^{\dagger}}$ must be smaller than $d$.

Therefore, it is not possible to have $\mathcal{D}_{\beta}^{\dagger}\times \mathcal{D}_{\alpha} = I + ...$ when $\alpha\neq \beta$ and the `dagger' of a non-invertible symmetry, which fuses with the symmetry with one channel being identity, is unique.
\end{proof}

We remark that these properties are close analogues of those satisfied by anyons in 2D topological states:
\begin{enumerate}
    \item Each anyons $\alpha$ in 2D topological states always have an inverse-anyon $\bar{\alpha}$ such that 
\begin{equation}
\alpha \times \bar{\alpha} = e + ...
\label{eq:anyon}
\end{equation}
where $\times$ represents anyon fusion and $e$ is the trivial anyon. $e$ shows up only once in the decomposition.
\item  The inverse of an anyon is unique.
\item  The fusion of two anyons can result in the sum of multiple anyons with non-negative integer coefficient.
\end{enumerate} 
This is of course nothing surprising because 
when realized in the Symmetry Topological Field Theory formalism, non-invertible 1D symmetries are represented by the string operator of anyons in a 2D topological state, which in general are implemented as sequential circuits. Therefore, for example, Eq.~\ref{eq:1Dnoninvert} and Eq.~\ref{eq:anyon} are manifestations of the same fact. We should emphasize, however, that we did not derive the structure of a fusion category. In fact, our formalism can also be used to describe continuous symmetries (both invertible and non-invertible), which is not captured by a fusion category.

\section{2D non-invertible symmetry with Cheshire twist}
\label{sec:Cheshire}

2D non-invertible symmetries can be non-invertible in a very different way. For example, for the non-invertible symmetry $\mathcal{C}$ with `Cheshire' symmetry twist discussed in this section, no symmetry can fuse with $\mathcal{C}$ such that one of the channels is identity. This is the property we call unannihilable. In the categorical language, it means the symmetry object does not have a dual. It is fundamentally related to the fact that not all 1D symmetry twists can be (pair) created `freely' from the vacuum\footnote{The fact that these defects are unannihilable does not contradict the fact that a closed loop of such a defect can be shrunken to the vacuum. However, this operation also requires a sequential circuit that scales with the size of the loop.}. Some 1D symmetry twists can be (pair) created from the vacuum using a 1D finite depth circuit. Others, like the Cheshire ones, can only be created from the vacuum using a 1D sequential circuit. Therefore, for the second type of non-invertible symmetries, the 2D sequential circuit that moves a symmetry twist across the bulk of the system does not contain the full information about the symmetry action. We also need to know the 1D sequential circuit that generates the symmetry twist from vacuum. In this section, we discuss the case of the non-invertible symmetry of 2D Toric Code with the `Cheshire string' symmetry twist. We will show how the full symmetry action can be derived once we know both the 2D sequential circuit for moving the twist and the 1D sequential circuit for generating the twist. The circuits were given in Ref.~\onlinecite{Tantivasadakarn2024} and we review them below. As shown below, the action of the symmetry $\mathcal{C}$ is to project to the symmetric subspace of a 1-form symmetry of the model. 

\begin{figure}[th]
    \centering
    \includegraphics[scale=0.45]{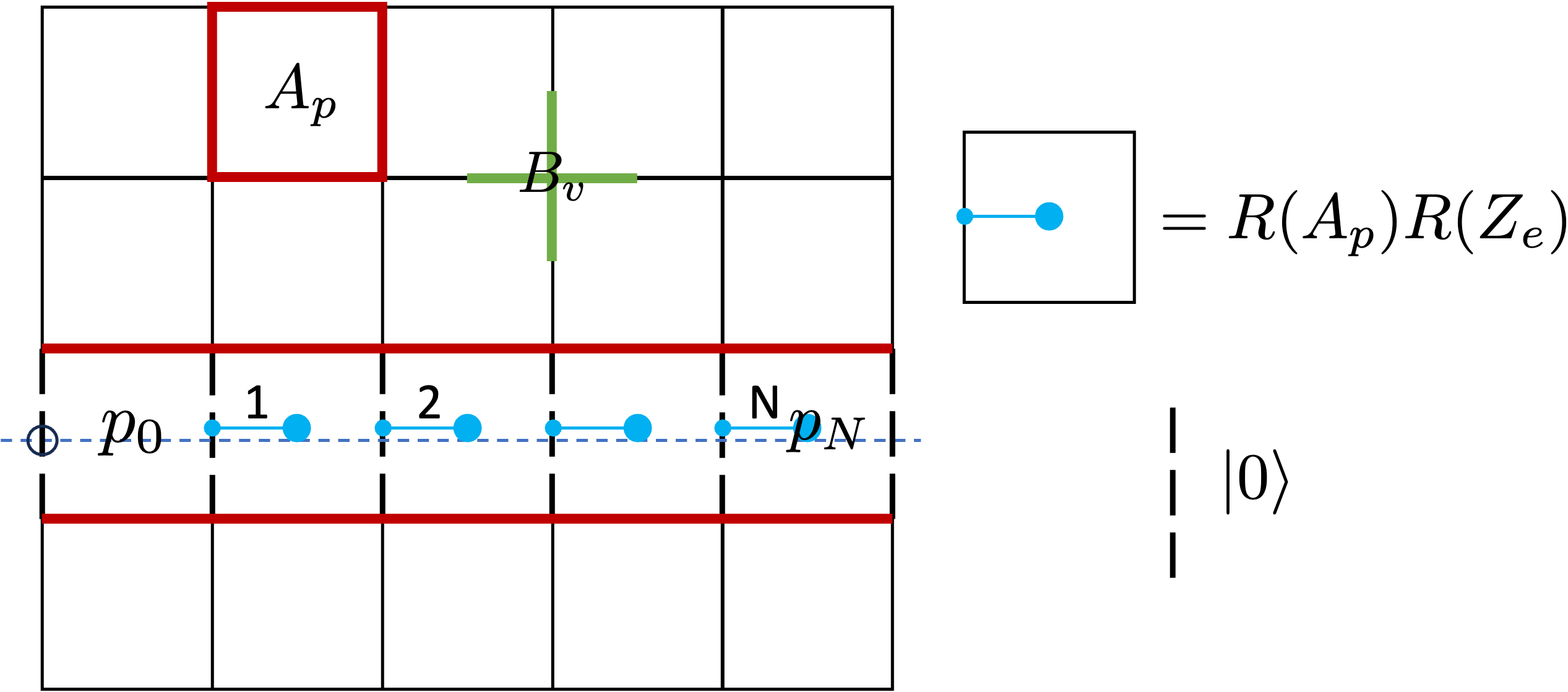}
    \caption{Generation of an $e$-Cheshire string in $2+1$D Toric Code with a sequential circuit. $A_p$ and $B_v$ are Hamiltonian terms of the $2+1$D Toric Code. A Cheshire string (on the dual lattice) from $p_0$ to $p_N$ is generated by applying a sequence of gate sets represented by the blue dot pairs. The dashed black edges are mapped to the product state $|0\rangle$ forming the condensate, while the total charge of the condensate measured by $\prod X$ around the red loop is conserved.}
    \label{fig:2DSe}
\end{figure}

Consider the Toric Code on a two dimensional square lattice with periodic boundary conditions (Fig.\ref{fig:2DSe}), defined with the Hamiltonian
\begin{equation}
  \begin{aligned}
H  = & -\sum_p A_p - \sum_v B_v \\
=&-\sum_p \prod_{e\in p} X_e - \sum_v \prod_{v \in e} Z_e
\end{aligned}  
\end{equation}
Let us call the excitation of the $A_p$ terms the gauge charge excitation labeled by $e$ and the excitation of the $B_v$ terms the gauge flux excitation labeled by $m$. The product of $Z_e$ along a closed loop in the dual lattice (including the $B_v$ terms) gives the 1-form $Z_2$ symmetry of this model. Among the 1-form symmetry operators, the ones that runs along nontrivial cycles in the $x$ and $y$ directions -- $W_x$ and $W_y$ -- are logical operators in the Toric Code ground space.

Applying $Z_e$ on one edge creates two gauge charge excitations on the neighboring plaquettes. Having a charge condensate corresponds to enforcing $-Z_e$ as the Hamiltonian term so that the ground state remains invariant under the pair creation or hopping of gauge charges between the neighboring plaquettes. If such a term is enforced on a string of edges on the dual lattice (dotted blue line in Fig.~\ref{fig:2DSe}), we get a Cheshire string for the gauge charge $e$.

Generating the Cheshire string with unitary operations requires a 1D sequential linear depth circuit that starts at plaquette $p_1$ and ends at plaquette $p_N$.
\begin{align}
   U= \prod_{i=1}^{N} R( A_{p_i})R(Z_{e_{i-1,i}})
\end{align}
where we define $R(\mathcal O) \equiv  e^{\frac{i\pi}{4} \mathcal O}$. 

After applying the circuit in $p_1$ through $p_N$, the string operator along the dotted line $\prod Z_e$ becomes equivalent to the $Z_e$ operator on the edge between $p_0$ and $p_N$ (marked by the open circle in Fig.~\ref{fig:2DSe}). Therefore, if we start from the ground state of the Toric Code with $\prod Z_e=1$ along the non-trivial horizontal cycle, we now have a complete Cheshire string. Otherwise, we would need to do a projection to complete the generation step of the Cheshire string symmetry twist.

\begin{figure}[ht]
    \centering
    \includegraphics[scale=0.43]{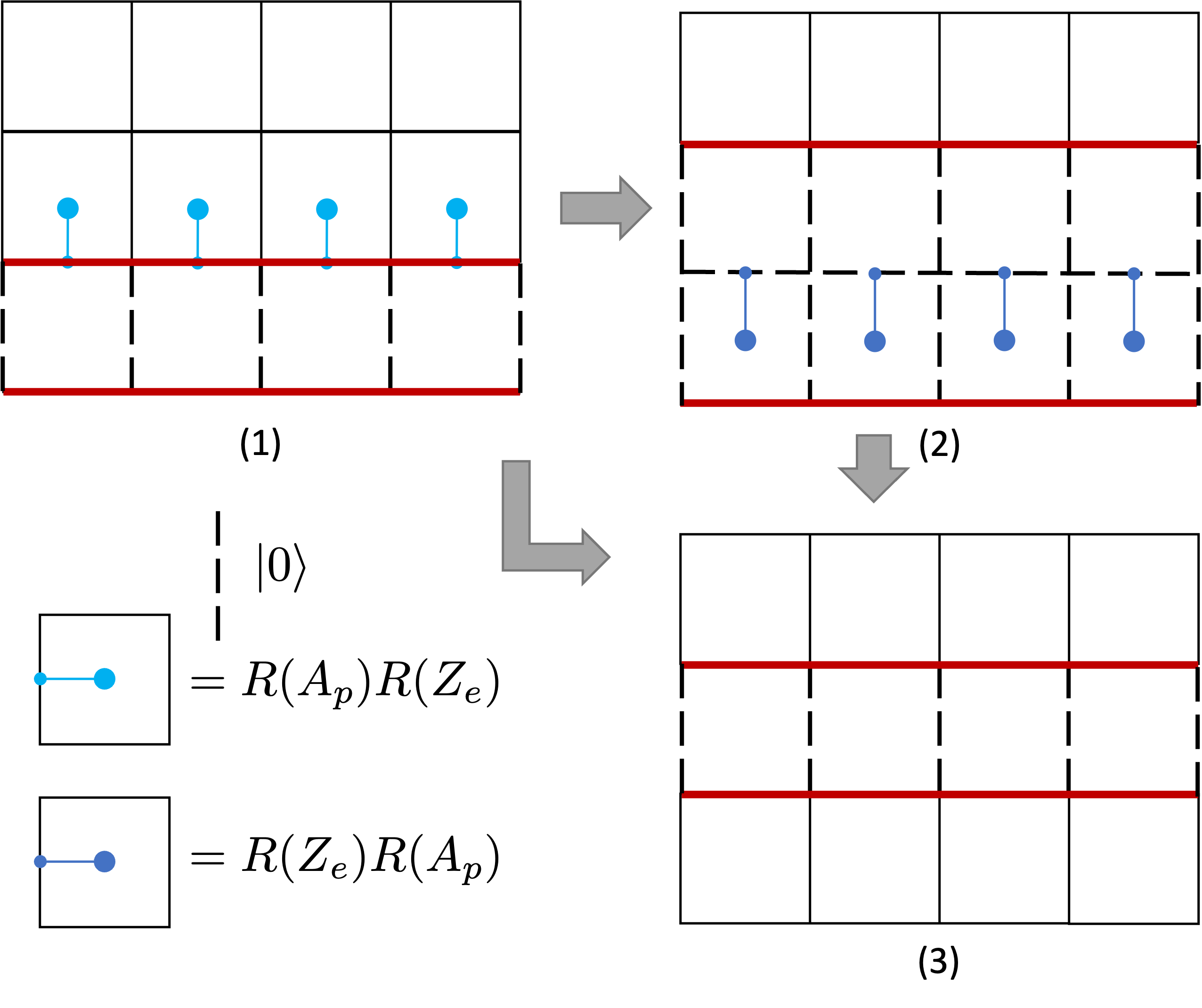}
    \caption{Deforming ((1) to (2) and (2) to (3)) and moving ((1) to (3)) a Cheshire string using a finite depth circuit in $2+1$D. The gate sets in each diagram can be applied in parallel. The dashed black edges are in the $|0\rangle$ state of the condensate. Red loops of $\prod X$ measure the total charge in the condensate.}
    \label{fig:2DSe_move}
\end{figure}

Once the symmetry twist is generated, we can move it with a 2D sequential circuit. As shown in Fig.~\ref{fig:2DSe_move}, if we start from a Cheshire string in the bottom row ((1) with dashed black edges in state $|0\rangle$), we can make it thicker (2) and then thinner (3) using finite depth circuits. The individual gate sets (the blue dot pairs) take the same (inverse) form as in Fig.~\ref{fig:2DSe}. The difference is that, now the gate sets are oriented in parallel, rather than connecting head to toe. It can be easily checked that parallel gate sets commute with each other and hence can be applied simultaneously. Going from (1) to (3) moves the Cheshire string perpendicular to its length by one step. The total charge of the condensate, measured by $\prod X$ along the red loop, is conserved in the whole process. To sweep the Cheshire string across the 2D plane, we can simply sequentially apply this finite depth circuit as the twist moves along.

Finally, when the symmetry twist returns to its original position, we can reverse the 1D sequential circuit in the generation step and remove the Cheshire string with a 1D sequential circuit. 

We emphasize that this symmetry is defined only in the ground space of the Toric Code (the low energy limit of a $Z_2$ gauge theory). It is not defined in the whole Hilbert space. A similar case in 1D is the non-invertible Fibonacci symmetry\cite{Feiguin2007} which is defined only on a constrained Hilbert space. It might be possible to extend the action of the symmetry onto the underlying un-constrained tensor product Hilbert space, but the fusion rules of the symmetry do not hold in the full Hilbert space and the symmetry cannot be applied with a sequential circuit anymore (i.e. its defect is no longer topological). In this sense, these symmetries are `emergent symmetries' of a particular gapped / gapless theory. See Ref.~\cite{Ravindran2026} for a more in-depth discussion of the notion of emergent symmetry.

\begin{figure}[ht]
    \centering
    \includegraphics[scale=0.65]{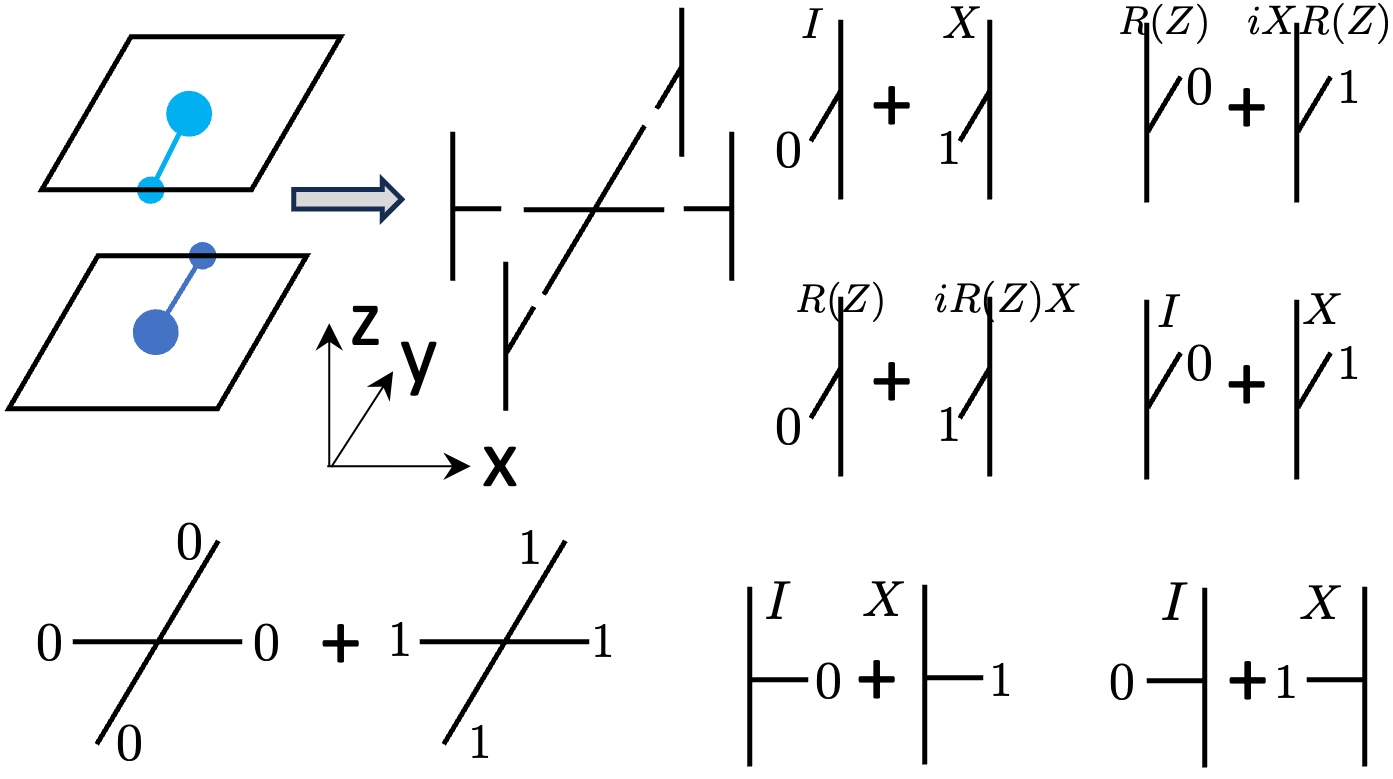}
    \caption{Decomposition of the gate sets in each plaquette (light and dark blue dot pairs) into five parts in the tensor network operator representation. The middle part and the $\pm x$ parts are the same for the two gate sets and are shown in the bottom part of the figure. The $\pm y$ parts are different for the two gate sets and are shown separately in the top right part of the figure.}
    \label{fig:C_TPOa}
\end{figure}

To obtain the full non-invertible symmetry action associated with the Cheshire string symmetry twist, we start from the 2D sequential circuit in the bulk that sweeps the Cheshire string. The gate set within each plaquette can be decomposed into a tensor network form as shown in Fig.~\ref{fig:C_TPOa}. Composing different steps in the sequential circuit together, we find the tensors on the horizontal edges to be

\bigskip
\includegraphics[scale=0.65]{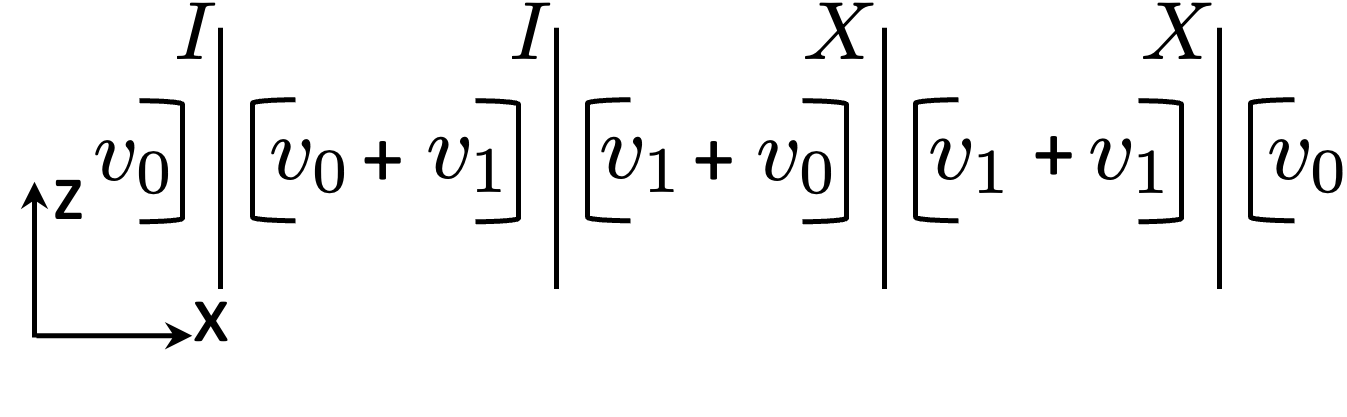},

the tensors on the vertical edges to be

\bigskip
\includegraphics[scale=0.65]{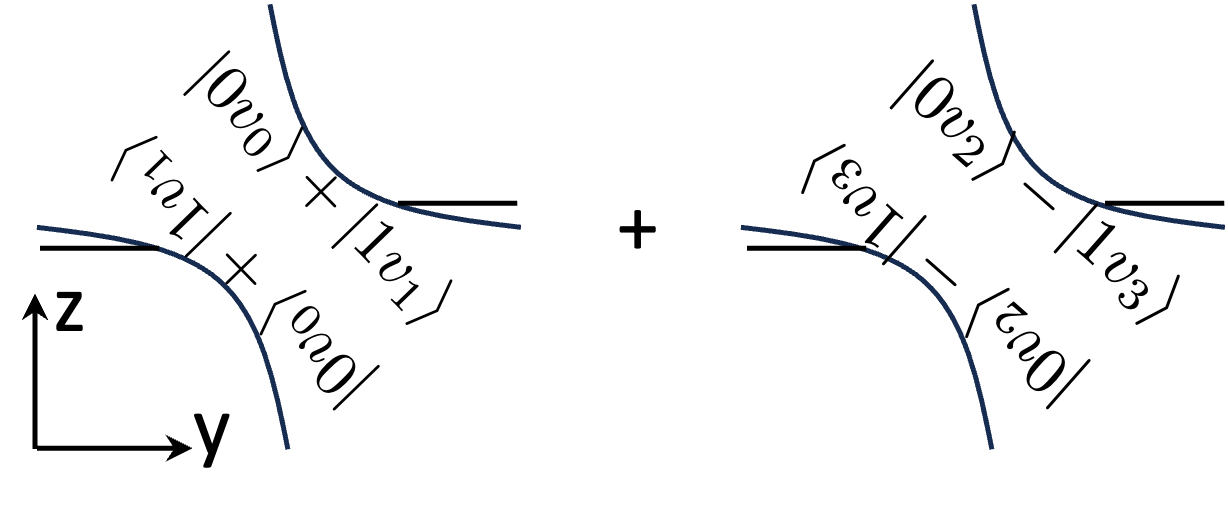},

and the tensor in each plaquette to be

\bigskip
\includegraphics[scale=0.65]{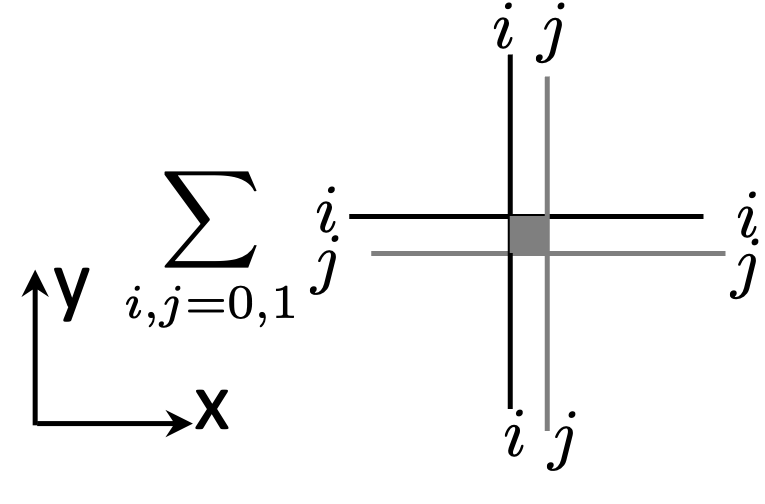}.

where $|v_0\rangle = \ket{00}+\ket{11}$, $|v_1\rangle = \ket{01}+\ket{10}$, $|v_2\rangle = \ket{01}-\ket{10}$, $|v_3\rangle = \ket{00}-\ket{11}$.

\begin{figure}[ht]
    \centering
    \includegraphics[scale=0.65]{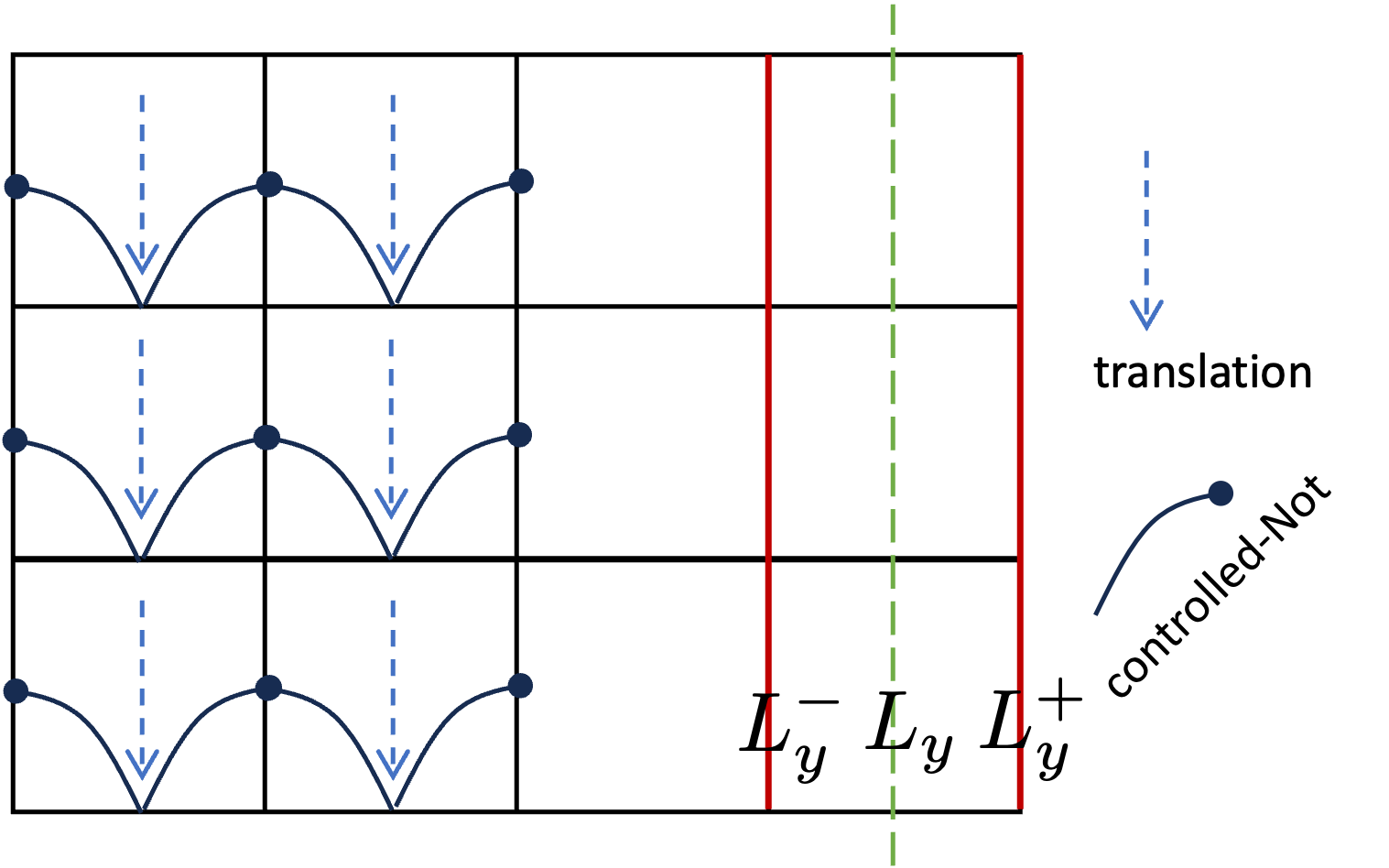}
    \caption{The operator resulting from the tensor network representation of the 2d sequential circuit that sweeps the Cheshire string in Toric Code. Besides the translation and controlled-Not operations, there is a projection on each column onto the eigenvalue $1$ sector of the product of the $X$'s along the red lines and the $Z$'s along the dashed green line.}
    \label{fig:C_2D}
\end{figure}

Now we can connect copies of this set of tensors to make a translation invariant TPO by taking periodic boundary condition in the $y$ direction. The overall operation of the resulting tensor network involves a translation of all horizontal edges by one step in the $-y$ direction, some controlled-Not operation from the horizontal edges to the vertical edges as indicated in Fig.~\ref{fig:C_2D} and, most importantly, a projection operator of 
\begin{equation}
I + \prod_{e\in L_y} Z_e \prod_{e \in L_y^-} X_e \prod_{e \in L_y^+} X_e 
\label{eq:py}
\end{equation}
with the $\prod Z$ string operator on the dashed green line and the two $\prod X$ string operators on the red lines. 

The translation and controlled-Not gates amounts to translation of all local operators symmetric under the 1-form $Z_2$ symmetry. The local symmetric operators come in three types: $Z_e$ on horizontal edges, $Z_e$ on vertical edges and $A_p$ on plaquettes. We can check that under translation and the controlled-Not gate, $Z_e$ on horizontal edges moves downward by one step, so does $A_p$. $Z_e$ on vertical edges is mapped to a three body $Z_eZ_eZ_e$ by adding to it $Z_e$s on the two horizontal edges right beneath it. Considering the $B_v = \prod_{v\in e} Z_e$ term at the same vertex, we see that $Z_e$ on vertical edges are effectively also translated downward by one step. Since one step translation becomes a trivial transformation in the continuum limit, we will ignore its effect.

The important part of the symmetry action is the projection in Eq.~\ref{eq:py}. $\tilde{W}_y = \prod_{e\in L_y} Z_e \prod_{e \in L_y^-} X_e \prod_{e \in L_y^+} X_e $ is a dressed 1-form symmetry operator. The product of two parallel strings of $\prod X$ along $L_+$ and $L_-$ is equivalent to the product of $A_p$ terms along $L$, which is equal to $1$ on the Toric Code ground space. Therefore, as part of a symmetry action on Toric Code, $\tilde{W}_y$ is equivalent to $W_y = \prod_{e\in L_y} Z_e$, the 1-form symmetry operator along the nontrivial cycle in the $y$ direction. Therefore, the projector in Eq.~\ref{eq:py} projects into the symmetric subspace of the 1-form symmetry operator along the nontrivial $y$ direction cycle. 

Next we will consider the 1D sequential circuit that generates and annihilates the Cheshire string. We will see that the translation-invariant MPO obtained from these two circuits gives a projection onto the symmetric subspace of $W_x$, the 1-form symmetry operator along the nontrivial $x$ direction cycle. 

The generation step involves the same local gate set used in the sweeping step. We choose the annihilation step to be the inverse of the generation step. Using the tensor network operator decomposition in Fig.~\ref{fig:C_TPOa}, we find that the combined generation - annihilation circuit can be represented with tensors on the horizontal edges given by

\smallskip
\includegraphics[scale=0.65]{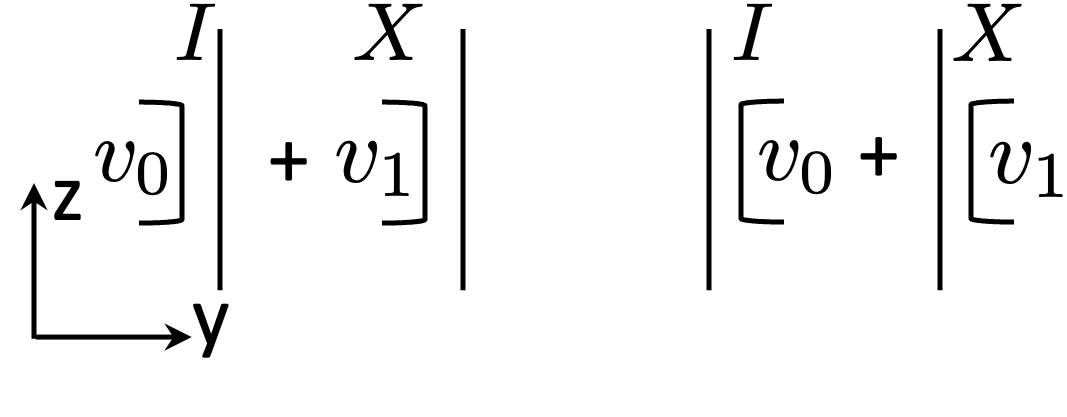},

the tensors on the vertical edges to be

\bigskip
\includegraphics[scale=0.65]{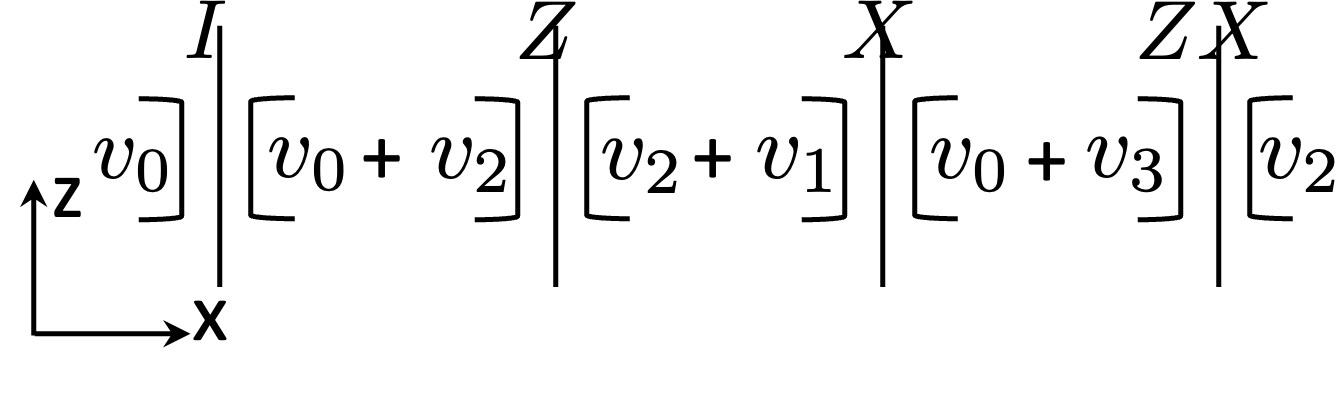},

and the tensor in each plaquette to be

\bigskip
\includegraphics[scale=0.65]{Tp2}.

Combining these tensors, we find the non-invertible symmetry action at the generation and annihilation step to be simply the same projection operator as Eq.~\ref{eq:py}, but oriented in the $x$ direction. 
\begin{equation}
I + \prod_{e\in L_x} Z_e \prod_{e \in L_x^-} X_e \prod_{e \in L_x^+} X_e = I + \tilde{W}_x
\end{equation}
which becomes equivalent to $I+W_x$ when acting on the Toric Code ground state.

The full symmetry action obtained from the 1D sequential circuit in the generation / annihilation step and the 2D sequential circuit in the sweeping step is a projection onto the symmetric space of the 1-form symmetry operator along the $x$ and $y$ nontrivial cycles. 
\begin{equation}
\mathcal{C} = \left(I+W_{x_0}\right)\prod_i\left(I+W_{y_i}\right)
\label{eq:C}
\end{equation}
where the product is over one row (labeled by $x_0$) and all columns (labeled by $y_i$). 

We obtain these particular projectors because we followed a particular pattern in generating (in $x$ direction) and sweeping (in $y$ direction) the symmetry twist. But the symmetry $\mathcal{C}$ is topological, meaning that we can choose any pattern for implementing the symmetry. Suppose that we start by generating a Cheshire string around the $x$ direction cycle but in a form locally deformed from the straight line, we would obtain a projector of the 1-form symmetry operator along the deformed line. Both the projector along the straight line and the deformed line should be part of the symmetry $\mathcal{C}$, which means their difference -- the projector onto a local $B_v$ term -- is part of the symmetry. In our previous derivation using the sequential circuit, this projector did not appear because we were already acting on the Toric Code subspace with $B_v=1$. 

If we combine the projection of local $B_v$ terms with those of nontrivial string operators $W_x$ and $W_y$, the total symmetry action becomes the projection onto the symmetric space of all 1-form operators, on both trivial and nontrivial cycles.
\begin{equation}
\mathcal{C}' = \sum_{\mathcal{L}} W_{\mathcal{L}}
\label{eq:C'}
\end{equation}
where $\mathcal{L}$ labels all closed loops on the dual lattice and $W_{\mathcal{L}}$ denote the 1-form symmetry operator around loop $\mathcal{L}$. When $\mathcal{L}$ shrinks to a point, $W_{\mathcal{L}}$ becomes the identity operator. 

Let us comment on the similarities and differences between the two forms of the full symmetry action $\mathcal{C}$ in Eq.~\ref{eq:C} and $\mathcal{C}'$ in Eq.~\ref{eq:C'}. As symmetries on the Toric Code ground state, $\mathcal{C}$ and $\mathcal{C}'$ have the same action -- they project onto the symmetric ground state of the 1-form symmetry. This naturally leads to the fusion rule
\begin{equation}
\mathcal{C}^{\dagger} \times \mathcal{C} \sim \mathcal{C}
\end{equation}
which also holds for $\mathcal{C}'$. On the other hand, while $\mathcal{C}$ can be implemented with sequential circuits in the bulk and measurements at the boundary, as we have derived in this section, $\mathcal{C}'$ contains local projections throughout the bulk and cannot be implemented this way. $\mathcal{C}'$ has the advantage of being more isotropic and is a short-range correlated operator, like all the generalized symmetry operators in 1D. 

To see that $\mathcal{C}'$ is short range correlated, we can think of the many-body operator with qubits on each edge as a many-body state with two qubits on each edge by combining the input and output indices of a qubit as the output indices of two qubits. In particular, the single qubit identity operator becomes a two qubit entangled state of $|I\rangle\rangle = v_0$ and the single qubit $Z$ operator becomes a two qubit entangled state of $|Z\rangle\rangle = v_3$. The many-body state obtained has a `loop-condensate' structure in the sense that the state is an equal weight superposition of closed loop configurations where $v_0$ on each edge labels the no string state and $v_3$ on each edge labels the string state. The many-body state is hence a modified version of the Toric Code state with the string and no-string states represented by two-qubit entangled states. Like the Toric Code state, $\mathcal{C}'$ is hence short-range correlated but long-range entangled.

With this understanding, it is straightforward to write a tensor network representation of $\mathcal{C}'$ by slightly modifying the tensor network representation of the Toric Code state. The representation contains tensors at vertices

\bigskip
\includegraphics[scale=0.65]{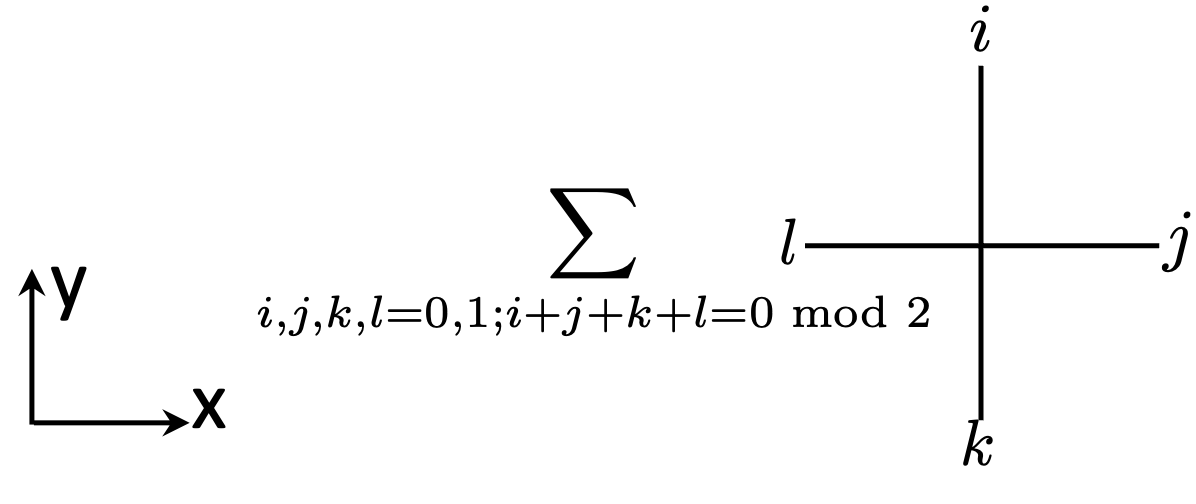},

as well as tensors on the edges.

\bigskip
\includegraphics[scale=0.7]{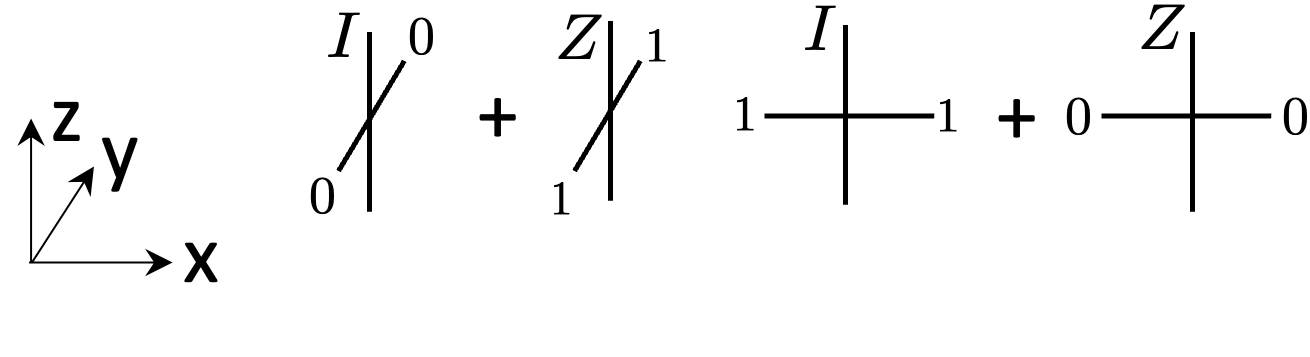}.

connected into a tensor network as shown,

\bigskip
\includegraphics[scale=0.6]{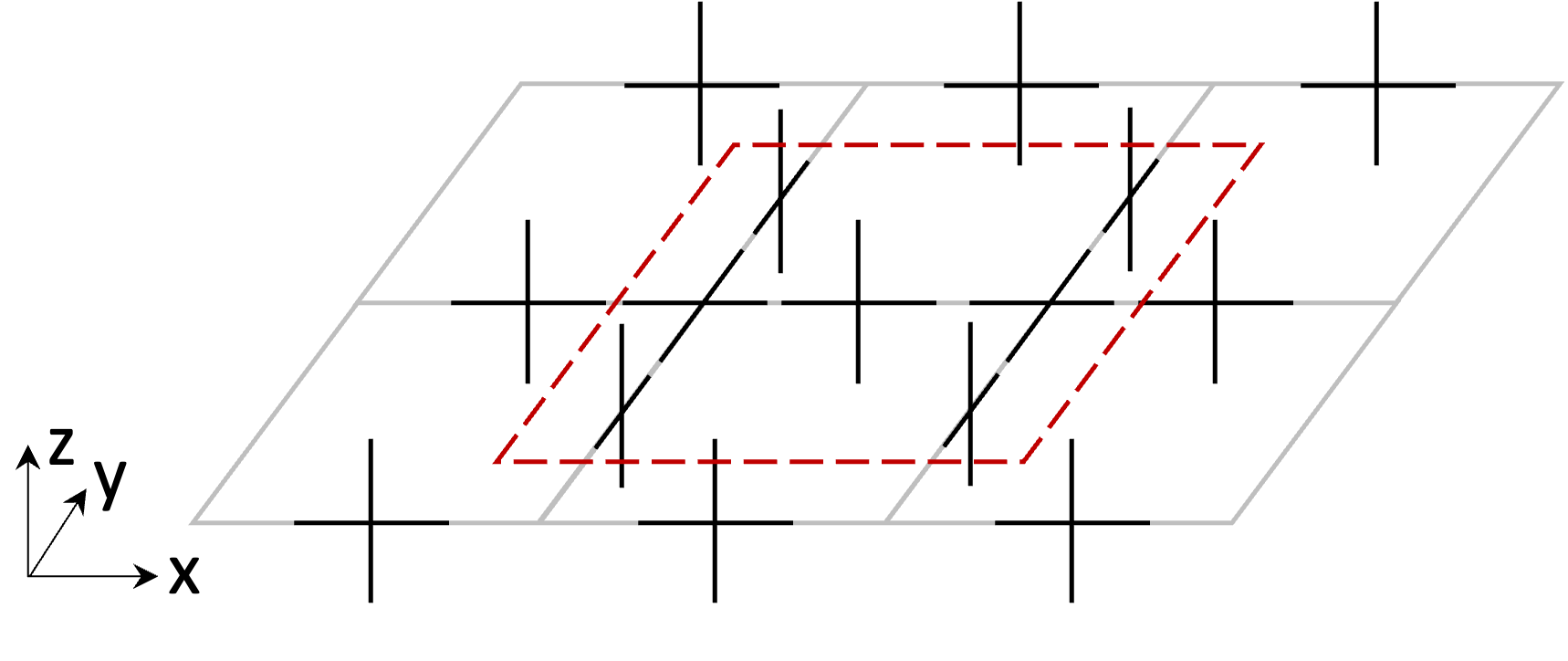}.

One prominent feature of this tensor representation is that, it has an internal $Z_2$ symmetry. Inserting a $Z$ operator on every horizontal index intersected by the red dashed loop in the figure above, we see that the TPO remains invariant. Because of this, the 2D TPO is very different from the 1D MPO discussed in section~\ref{sec:1D} which are all injective. The 2D TPO is not injective, but rather only $Z_2$-injective, a prominent feature of the tensor network state representation of string-net wave functions\cite{Swingle2010,Schuch2010}. This is the extra feature that can appear in 2D tensor network operators and relatedly 2D generalized symmetries. 

\section{Outlook}

In this paper, we explored the relation between generalized symmetry and sequential circuit on the lattice, which can be summarized as follows. First, the topological nature of the generalized symmetries dictates that the bulk of the symmetry is implemented as sequential circuits. Secondly, while the sequential circuit cannot represent the full symmetry action which can be non-invertible, it contains the necessary information to construct the full symmetry action. 

To illustrate the second statement, we studied two cases in detail. First, for 1D generalized symmetries, the full symmetry action can be obtained from the sequential circuit that sweeps the symmetry twist in the bulk by forming a translation-invariant form of the matrix product operator representation of the sequential circuit. Using the fact that the matrix product representation of the generalized symmetry comes from a sequential circuit, we can prove many properties of 1D generalized symmetries including the fact that they are all annihilably non-invertible, meaning that the fusion of the symmetry operator $\mathcal{D}$ and its complex conjugate $\mathcal{D}^{\dagger}$ contains one and only one identity channel.
\begin{equation}
\mathcal{D}^{\dagger} \times \mathcal{D} = I + ...
\end{equation}
%\xc{\sout{Our proof applies to generalized symmetries that act on a tensor product Hilbert space, like the Kramers-Wannier example. It is known that some generalized symmetries act only on a constrained Hilbert space, like the Fibonacci symmetry. To generalize our proof to these cases, we need a canonical form for the matrix product state on a constrained Hilbert space. We leave such a proof to future study, although we expect all conclusions (claims~\ref{claim1},\ref{claim2},\ref{claim3},\ref{claim4}) to hold.} 
Our proofs in section~\ref{sec:1D} apply to symmetries acting on a tensor product Hilbert space. In Appendix~\ref{app:constrained}, we show how the proofs can be generalized to non-tensor product constrained Hilbert space.

Next, we studied a 2D non-invertible symmetry which is unannihilable, that is, there does not exist another symmetry which fuses with it and contains identity in one of the channels. To obtain the full symmetry action, we need not only the 2D sequential circuit that sweeps the symmetry twist but also the 1D sequential circuit that generates and annihilates the symmetry twist. With periodic boundary conditions, the tensor / matrix product representation of the 2D / 1D sequential circuits gives us the full non-invertible symmetry action. We find that an isotropic version of the symmetry operator is a toric-code-like many-body operator, which is short-range correlated but long-range entangled, a feature not possible in 1D. 

We did not systematically study 2D generalized symmetries in this paper, as we did for 1D. Presumably, the 2D generalized symmetries are represented by short-range correlated translation-invariant tensor network operators that come from sequential circuits. In 1D, based on such a setup, we were able to establish many common features of generalized symmetries. In 2D, we already see that there are more possibilities. It will be interesting to see what can be proven in general.

On a more specific note, it might be interesting to think about `half' of the symmetry $\mathcal{C}$. The sequential circuit implementing $\mathcal{C}$ contains two types of gate sets -- the light blue one and the dark blue one (as shown in Fig.~\ref{fig:2DSe} and \ref{fig:2DSe_move}) -- and maps from Toric Code back to Toric Code. If we use only the light blue gate sets, it maps the Toric Code to the Higgs phase and induces a transition. The circuit actually has a lot of similarity with the Kramers-Wannier circuit in 1D. Whether the circuit gives rise to a useful generalized symmetry and how the generalized symmetry might constrain the critical point of the Higgs transition is an interesting problem to look into.

%%%%%%%%%%%%%%%%%%%%%%%%%%%%%%%%%%%%%%%%%%%%%%%%%%%

\begin{acknowledgments}
We are indebted to inspiring discussions with Shu-Heng Shao, Sahand Seifnashri, Nathan Seiberg, Daniel Ranard, Robijn Vanhove, Frank Verstraete, Clay Cordova, Sakura Schafer-Nameki and Paul Fendley. We are particularly grateful to Linqian Wu for helping to coin the term `annihilable', and the anonymous referee for encouraging us to generalize our results to constrained Hilbert spaces. X.C. is supported by the Simons collaboration on `Ultra-Quantum Matter'' (Grant No. 651438), the Simons Investigator Award (award ID 828078) and the Institute for Quantum Information and Matter at Caltech. N.T. and X.C. are supported by the Walter Burke Institute for Theoretical Physics at Caltech. X.L. is supported by the David and Barbara Groce Fellowship at Caltech. 
\end{acknowledgments}
\newpage

\bibliography{references}

\appendix

\section{Canonical form of matrix product state}
\label{app:MPS}

In this section, we review the canonical form of the matrix product states (MPS) with periodic boundary condition derived in Ref.~\onlinecite{Perez-Garcia2007}. For the discussion in this paper, it is important to know not only the statement of the canonical form theorem, but also how the canonical form is derived. We review the steps taken to obtain the canonical form of a matrix product state in a constructive manner. To apply the result to matrix product operators, which is the center of discussion in the main text, we mention some minor modifications at the end.

Given a translationally invariant MPS with bond dimension $D$ on a 1D periodic lattice defined by the $D\times D$ representation matrices $A_i$, where $i$ is the index of on-site physical Hilbert space basis, we can define the transfer matrix operator 
\begin{equation}
\mathbb{T}(X)=\sum_i A_i X A_i^\dagger
\end{equation}

Due to the `completely-positive' property of $\mathbb{T}$, it can be shown that $\mathbb{T}$ has a positive fixed point. That is,
\begin{equation}
\mathbb{T}(X)=\sum_i A_i X A_i^\dagger=\lambda^2 X
\end{equation}
where $\lambda^2 > 0$ is the largest (in magnitude) eigenvalue of $\mathbb{T}$ and $X$ is a positive matrix. Here, $\lambda^2$ is called the spectral radius of $\mathbb{T}$.

Note that it is possible that $\mathbb{T}$ has more than one fixed points. Among all the fixed points of $\mathbb{T}$, choose one with the minimum support (without loss of generality, let us still label it as $X$). That is, no other fixed point $X'$ has a support strictly contained in the support of $X$. If $X$ is invertible (i.e. supported in the full space), then $B_i=\lambda^{-1} X^{-\frac{1}{2}} A_i X^{\frac{1}{2}}$ satisfies $\sum_i B_i B_i^{\dagger}=\mathbb{I}$ and $\mathbb{I}$ is the only fixed point of the transfer matrix of $B_i$. $\lambda B_i$ gives the canonical form of the MPS. In this case, the MPS has only one block in its canonical form and is called `injective'. 

If $X$ is not invertible, denote by $P$ the projection onto the support space of $X$. It can be shown that
\begin{equation}
A_i P = P A_i P.
\end{equation}
Therefore, if $P_{\perp}$ denotes the projection onto the orthogonal subspace, in the block basis spanned by $P$ and $P_{\perp}$, the matrix \( A_i \) takes the block upper triangular form
\begin{equation}
A_i = \begin{bmatrix}
PA_iP & P A_i P_{\perp} \\
0 & P_{\perp} A_i P_{\perp}
\end{bmatrix}.
\end{equation}
The transfer matrix of the first diagonal block $PA_iP$ has a unique full-rank (within $P$) fixed point $X$, similar to the injective case. Defining $B_i = \lambda^{-1}X^{-1/2}PA_iPX^{1/2}$, $C_i = P_{\perp}A_iP_{\perp}$, $F_i = X^{-1/2}PA_iP_{\perp}$,  we can rewrite the original MPS in terms of the following $D\times D$ matrices
\begin{equation}
A_i = 
\begin{bmatrix}
\lambda B_i & F_i \\
0 & C_i
\end{bmatrix}\ .
\end{equation}
where the first diagonal block is in the canonical form having $\mathbb{I}$ as the only fixed point $\sum B_iB_i^{\dagger} = \mathbb{I}$. We can repeat the procedure for $C_i$, so that eventually the $A_i$'s can be decomposed into a block upper triangular form
\begin{equation}
A_i = \left(
\begin{array}{cccccc}
\lambda_1 A_i^{1} & F_i^{1, 2} & F_i^{1, 3} & \cdots \\
     0    & \lambda_2 A_i^{2} & F_i^{2, 3} & \cdots \\
     0    &     0     & \lambda_3 A_i^{3} & \cdots \\
     0    &     0     &     0     & \ddots
\end{array}
\right)
\label{eq:MPScano}
\end{equation}
where $\lambda_1^2\geq\lambda_2^2\geq \lambda_3^2 \geq ... >0$ are the spectral radii of each diagonal block and the matrices $A_i^m$ in each block satisfy the conditions:

1. $\sum_i A_i^m A_i^{m \dagger}=\mathbb{I}$.

2. $\sum_i A_i^{m \dagger} \Lambda^m A_i^m=\Lambda^m$, for some diagonal positive and full-rank matrices $\Lambda^m$.

3. $\mathbb{I}$ is the only fixed point of the operator $\mathbb{T}_m(X)=\sum_i A_i^m X A_i^{m \dagger}$.

Note that replacing all the upper triangular blocks $F_i^{m,m'}$ by $0$ does not change the physical wavefunction of the MPS.

In our calculations in this paper, we work with Matrix Product Operator (MPO) representations, which can be mapped into MPS by treating the pair of physical indices \( ii' \) of the MPO tensor \( M \) as a single combined physical index \( \tilde{i}=(ii') \) for the tensor \( A_{\tilde{i}} \), i.e.,
\begin{equation}
    A_{\tilde{i}} = M_{ii'}.
\end{equation}
With this identification, the standard canonical decompositions for translation-invariant Matrix Product State above can be applied to translation-invariant matrix product operators.

For MPS, the overall normalization is not important and we are free to re-scale the matrices. For the sequential MPO discussed in this paper, the normalization is directly associated with the fact that the MPO comes from a sequential unitary circuit and cannot be changed arbitrarily. In particular, for injective sMPOs, we will use the normalization
\begin{equation}
    \sum_{ii'} M_{ii'} \left(M_{ii'}\right)^\dagger = d \, \mathbb{I},
\end{equation}
where \( d \) denotes the physical Hilbert space dimension per site, as specified in Equation~(\ref{eq:dt}).

For non-injective sMPOs, each diagonal block satisfies the above normalization.

\section{Proof for Claims 1 to 4 in constrained Hilbert spaces}
\label{app:constrained}
Many interesting 1D non-invertible symmetries are supported on constrained Hilbert spaces. For example, the Fibonacci symmetry~\cite{Feiguin2007} $D_{\tau}$ with fusion rule 
$$
\CD_{\tau} \times \CD_{\tau} = \CD_1 + \CD_{\tau}
$$ 
has quantum dimension $d_{\tau} = \frac{1+\sqrt{5}}{2}$. In fact, it was recently proven that such symmetries cannot be realized on a tensor product of finite dimensional Hilbert spaces~\cite{evans2026operator,jones2026structurec}. In this section, we show how the proofs in Section~\ref{sec:1D} can be generalized to apply to such symmetries as well.

First, we need to point out an important condition the constrained Hilbert needs to satisfy: the projection operator on the constrained Hilbert space can be written as an injective MPO. This condition is satisfied by, for example, all the anyon chains~\cite{Feiguin2007}. In 1D systems like the anyon chains, the Hilbert space does not have a tensor product structure. Rather, it can be obtained from a 1D tensor product Hilbert space by imposing local constraints. Graphically, the local constraints can be represented as shown in Fig.~\ref{fig:constrain}, which naturally gives rise to an MPO representation of the projector -- the unit operator -- onto the Hilbert space. Our following discussion applies when this unit MPO is injective. In Lemma~\ref{lemma1}, we show that if the MPO projector is not injective, it decomposes into the direct sum of injective MPO projectors and without loss of generality we can study the generalized symmetry acting on each projected space separately. With this Lemma, we will assume in the following proofs that the MPO projector is injective. This rules out cases like the ground space of symmetry breaking phases
\begin{equation*}
    P = \ket{00...0}\bra{00...0} + \ket{11...1}\bra{11...1}
\end{equation*}

\begin{figure}[h]
    \centering
    \includegraphics[scale=0.55]{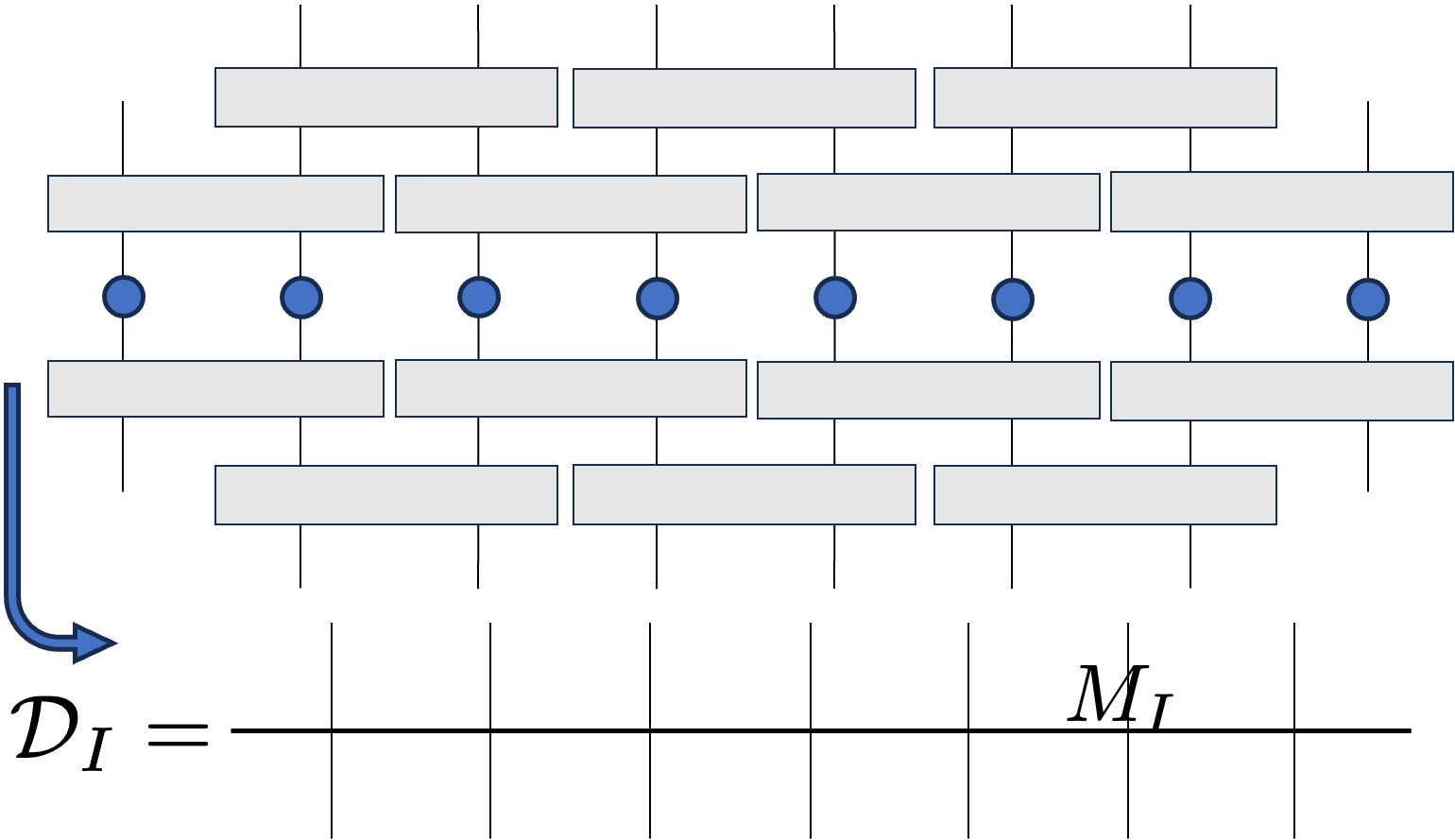}
    \caption{The MPO representation of the projector onto the constrained Hilbert space. Blue dots represent local identity operators on underlying tensor product Hilbert space. Gray boxes are local constraints.}
    \label{fig:constrain}
\end{figure}

For the Fibonacci anyon chain, we can see that this condition is satisfied. The unit operator on the Fibonacci chain Hilbert space is given by the MPO operator~\cite{molnar2022matrix}
\begin{figure}[h!]
    \centering
\includegraphics[scale=0.7]{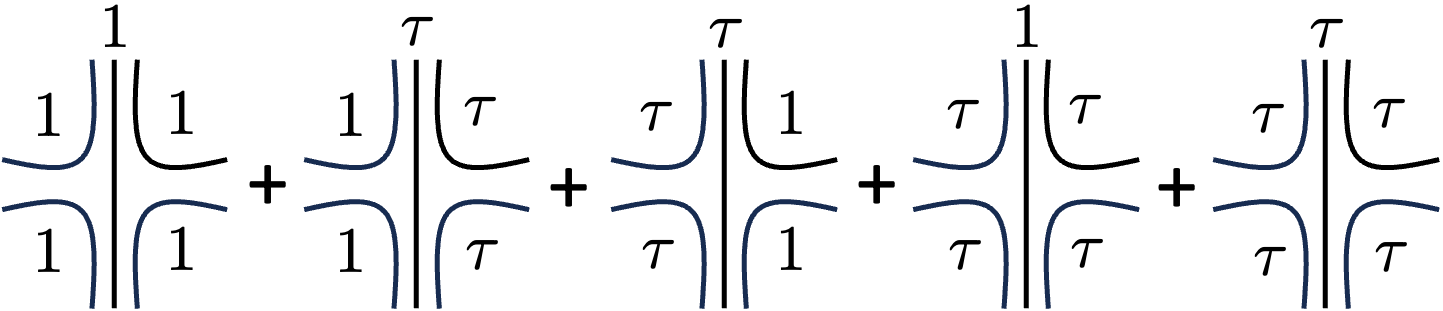}
    \label{fig:M1_Fib}
\end{figure}
where each index has two dimensions labeled by $1$ and $\tau$. Indices at the two ends of the same line share a label ($1$ or $\tau$). It can be checked that this is a projection operator and an injective MPO. The transfer matrix has a non-degenerate largest eigenvalue $d_1 = (3+\sqrt{5})/2$. The trace of this operator on a chain of length $N$ is given by
$$
\text{Tr} (\CD_1) = \left( \frac{3+\sqrt{5}}{2}\right)^N + \left( \frac{3-\sqrt{5}}{2}\right)^N
$$
which approaches $d_1^N$ in the limit $N\to \infty$. Therefore, $d_1$ is the effective local dimension in the chain.

What is a sequential circuit on the constrained Hilbert space? Consider first a sequential circuit $\tilde{U}$ on the underlying tensor-product Hilbert space. Applying the sequential circuit to the constrained Hilbert space could violate the local constraints on the Hilbert space. For the circuit to be a sequential circuit on the constrained Hilbert space, we require that only local constraints near the two ends of the circuit can be violated. Constraints in the middle of the circuit are preserved. In particular, once we impose periodic boundary condition on the MPO of the sequential circuit, the resulting operator $\tilde{\CD}$ preserves all the local constraints and satisfies
\begin{align}
    \tilde{\CD}\CD_1 = \CD_1\tilde{\CD}.
    \label{eq:UD}
\end{align}

The generalized symmetry on the constrained Hilbert space is then given by the operator
\[
\CD = \tilde{\CD}\CD_1 = \CD_1\tilde{\CD} = \CD_1\tilde{\CD}\CD_1
\]
The MPO representation of $\CD$ can be obtained by composing the MPO of the sequential circuit and the MPO of $\CD_1$.

\setcounter{lemma}{0}
\begin{lemma}
If an MPO $\CD_1$ is a projector that satisfies $\CD_1^{\dagger} = \CD_1$, $\CD_1^2=\CD_1$ and decomposes into two injective components $\CD_{1a}$ and $\CD_{1b}$ ($\CD_{1a} \neq \CD_{1b}$), then $\CD_{1a}$ and $\CD_{1b}$ are orthogonal projectors. That is
\[
\begin{array}{l}
\CD_{1a}^{\dagger} = \CD_{1a}, \ \CD_{1b}^{\dagger} = \CD_{1b},\\
\CD_{1a}^2 = \CD_{1a}, \ \CD_{1b}^2 = \CD_{1b}, \ \CD_{1a}\CD_{1b} = 0
\end{array}
\]
\label{lemma1}
\end{lemma}

\begin{proof}
Since $\CD_1$ as an MPO contains two injective components $\CD_{1a}$ and $\CD_{1b}$, $\CD_1^2$ contains four MPO components
\[
\CD_{1a}^2,\ \CD_{1b}^2, \ \CD_{1a}\CD_{1b}, \ \CD_{1b}\CD_{1a}
\]
Each of these components is not necessarily injective and can be decomposed into injective sub-components. Since $\CD_1^2 = \CD_1 = \CD_{1a}+\CD_{1b}$, among all the injective sub-components, only two are nonzero: one gives rise to $\CD_{1a}$ and the other gives rise to $\CD_{1b}$. Therefore, among the four MPO components, either (a) one of them is equal to $\CD_{1a}+\CD_{1b}$ and the other three are zero, or (b) two of them are separately equal to $\CD_{1a}$ and $\CD_{1b}$ and the remaining two are zero. We will find out that only case (b) is possible.

Since $\CD_1^{\dagger} = \CD_1$, there are two possible relations between $\CD_{1a}$ and $\CD_{1b}$. Either $\CD_{1a}^{\dagger} = \CD_{1a}$, $\CD_{1b}^{\dagger} = \CD_{1b}$ or $\CD_{1a}^{\dagger} = \CD_{1b}$, $\CD_{1b}^{\dagger} = \CD_{1a}$. 

First we will rule out the possibility of $\CD_{1a}^{\dagger} = \CD_{1b}$, $\CD_{1b}^{\dagger} = \CD_{1a}$. Suppose that this is the case. Since $(\CD_{1a}^2)^{\dagger} = \CD_{1b}^2$, $\CD_{1a}\CD_{1b}$ and $\CD_{1a}\CD_{1b}$ are (Hermitian and) positive semi-definite, which contradicts with option (a). To satisfy option (b), $\CD_{1a}\CD_{1b}$ and $\CD_{1a}\CD_{1b}$ have to be zero because $\CD_{1a}$ and $\CD_{1b}$ are not Hermitian. But $\CD_{1a}\CD_{1b} = \CD_{1a}\CD_{1a}^{\dagger}$ is positive semi-definite and cannot be zero as long as $\CD_{1a} \neq 0$.

Now consider the possibility that $\CD_{1a}$ and $\CD_{1b}$ are both Hermitian. As long as $\CD_{1a}$ and $\CD_{1b}$ are both nonzero, $\CD_{1a}^2$ and $\CD_{1b}^2$ are nonzero. Therefore, option (a) is not possible. For option (b), since $\CD_{1a}^2$ and $\CD_{1b}^2$ are nonzero, we can have only two possibilities: (1) $\CD_{1a}^2 = \CD_{1a}$, $\CD_{1b}^2 = \CD_{1b}$, or (2) $\CD_{1a}^2 = \CD_{1b}$, $\CD_{1b}^2 = \CD_{1a}$, and in both cases $\CD_{1a}\CD_{1b} = \CD_{1b}\CD_{1a}=0$. But option (2) is not possible because $\CD_{1b}\CD_{1a} = \CD_{1b}^3$ and cannot be zero unless $\CD_{1b}$ is zero. This leaves us with the only possibility that $\CD_{1a}$ and $\CD_{1b}$ are Hermitian orthogonal projectors.
\end{proof}

Now consider a generalized symmetry acting on $\CD_1 = \CD_{1a}+\CD_{1b}$.  It either preserves both $\CD_{1a}$ and $\CD_{1b}$ or maps them into (linear combinations of) each other. In the first case, we can consider the symmetry action in $\CD_{1a}$ and $\CD_{1b}$ separately. In the second cases, the symmetry is spontaneously broken and we should focus instead on the unbroken symmetries acting separately on $\CD_{1a}$ and $\CD_{1b}$. Therefore, in the following proofs, we will assume without loss of generality that the unit MPO is injective.

\begin{comment}
An operators acting on this constrained Hilbert space is said to be local if it can be extended into a local operator acting on the underlying tensor product Hilbert space. The operator is said to be unitary if it is unitary within the local constrained Hilbert space. That is, it satisfies %without loss of generality, we assume the extended operator is zero outside the constrained Hilbert space. 
%That is, for a local operator $U$ acting on the underlying tensor product Hilbert space, $U$ is unitary on the contrained Hilbert space if it satisfies
\begin{equation}
U\CD_I = \CD_I U, \ \text{and} \ \ UU^{\dagger} \CD_I = \CD_I
\label{eq:UD}
\end{equation}
With this definition, we can study generalized symmetries acting on the constrained 1D system. The generalized symmetries can be implemented as a sequential circuit on the constrained Hilbert space, composed of local unitary operators as defined above, 
%the underlying tensor product Hilbert space where each local unitary satisfies Eq.~\ref{eq:UD}, 
followed by a proper projection of the ancillas to fix the boundary condition. The sequential circuit gives rise to an MPO acting on the constrained Hilbert space, which can be obtained by composing the MPO of the sequence of local unitary operators (as given in Fig.~\ref{fig:SC2MPO}) and the MPO of $\CD_I$ (as given in Fig.~\ref{fig:constrain}). 
\end{comment}

Now let us discuss what needs to be modified in the proofs of Claims~\ref{claim1} through \ref{claim4} to apply to the general case. 

The proof of Claim~\ref{claim1} applies directly. This is because, to implement a generalized symmetry on the constrained Hilbert space, we can without loss of generality apply a sequential circuit to the underlying tensor product Hilbert space (and project the pair of ancilla's to close the boundary). As long as the overall operation preserves the constrained Hilbert space (by satisfying Eq.~\ref{eq:UD}), we are effectively applying a symmetry to the constrained Hilbert space. Moreover, as long as the sequential circuits for $\CD_{\alpha}$ and $\CD_{\beta}$ violate local constraints only at the two ends, the composite circuits for $\CD_{\alpha}\times \CD_{\beta}$ and $a\CD_{\alpha}+b\CD_{\beta}$ also satisfy this condition. Therefore, the arguments in the proof of Claim~\ref{claim1} go through without needing any modification.

As discussed above, generalized symmetries on a constrained Hilbert space can still be represented as MPOs, therefore, the proof of Claim~\ref{claim2} also goes through. In particular, the canonical form of MPS can still be applied to MPOs following the mapping in Eq.~\ref{eq:M2A}, whether the MPOs act on constrained Hilbert space or not. `Simple' symmetries are still represented by `injective' MPOs. Note our requirement that the projection MPO is injective, hence the unit symmetry is a simple symmetry, as expected. When two MPOs $M_{\alpha}$ and $M_{\beta}$ fuse as illustrated in Eq.~\ref{eq:FusionMPO}, the composite MPO may no longer be injective and may contain more than one diagonal block $M_{\gamma}$ in its canonical form. To show that each diagonal block is a sequential MPO, we can use the same proof for Claim~\ref{claim2}. % because the sequential circuit of the component symmetries are sequential circuits acting on the underlying tensor product Hilbert space. 
By composing the sequential circuits of $\CD_{\alpha}$ and $\CD_{\beta}$ and choosing properly the projection of ancilla's at the end, as discussed in the proof of Claim~\ref{claim2}, we obtain the implementation of $\CD_{\gamma}$. 

A very important parameter in this proof is the normalization of each $M_{\gamma}$. When $M_{\gamma}$ acts on a tensor product Hilbert space, we showed that its normalization is such that $M_{\gamma}$ has spectral radius $d$, the dimension of the local Hilbert space. When $M_{\gamma}$ acts on a constrained Hilbert space, we will see that $M_{\gamma}$ has a spectral radius of $d_1$, the effective local dimension of the 1D chain. Once the normalization of each $M_{\gamma}$ is fixed, we see that there can be no prefactors $|\lambda_{\gamma}|\neq 1$ in front of $M_{\gamma}$ in the decomposition. As a result, the fusion of $\CD_{\alpha}$ and $\CD_{\beta}$ decomposes into a sum of $\CD_{\gamma}$ with non-negative integer coefficients, even when they act on a constrained Hilbert space.

The normalization of an sMPO $M_{\alpha}$ comes from the proof of Claim~\ref{claim3}. The proof of claim~\ref{claim3} boils down to showing that in the decomposition of $\mathbb{M}_{\alpha}$ (Eq.~\ref{eq:Ma1}), there is one and only one component of unit, which in the constrained Hilbert space is represented by $M_1$ (Fig.~\ref{fig:constrain}). The existence of at least one $M_1$ can be argued in a similar way to the proof of claim~\ref{claim3}. The fact that $M_{\alpha}$ comes from a sequential circuit means that by properly choosing the boundary condition, we should be able to represent $\CD_1$ using $\mathbb{M}_{\alpha}$, hence the existence of at least one block that is $M_I$. Because of the existence of this block, we know that $\mathbb{T}_{\alpha}$ has an eigenvector with eigenvalue $d_1$.

The fact that there can only be one block has to do with the normalization of $M_{\alpha}$. Since $\CD^{\dagger}_{\alpha} \times \CD_{\alpha}$ can be implemented as a sequential circuit and its dagger together with non-unitary operations at the end, its trace on the constrained Hilbert space can only differ from that of $\CD_1$ by a constant factor. So Eq.~\ref{eq:traceDD} and \ref{eq:Ta} can still be used to argue that $d_1$ is the largest eigenvalue of $\mathbb{T}_{\alpha}$ with the minor difference that $\text{Tr}(\CD_1) = d_1^N$ only in the limit of $N\to \infty$. Now using the condition that $M_{\alpha}$ is injective, we see that the $d_1$ has to be a non-degenerate eigenvalue. Therefore, there cannot be another diagonal block $M_I$ in the decomposition of $\mathbb{M}_{\alpha}$. As a corollary, we see that the normalization of each injective sMPO is such that its spectral radius is equal to $d_1$.

Finally, let us look at the proof of Claim~\ref{claim4}. The proof uses the fact that the inner product of two different MPS's is exponentially small when the MPS's are normalized. This is true for any pair of MPS's whether they are obtained from MPO acting on a tensor product Hilbert space or a constrained Hilbert space. Therefore, the argument in the proof still applies as long as we replace $d$ with $d_1$, the effective local dimension of the 1D chain.

\end{document}